\documentclass[a4paper,10pt]{article}
\usepackage{amsmath,color}
\usepackage{graphicx}
\usepackage{amsmath}
\usepackage{amsthm}
\usepackage{amssymb}
\usepackage[font=small,labelfont=bf]{caption}
\newtheorem{theorem}{Theorem}
\newtheorem{hypo}[theorem]{Hypothesis}
\newtheorem{prop}[theorem]{Proposition}
\newtheorem{lemma}[theorem]{Lemma}
\newtheorem{fact}[theorem]{Fact}

\newtheorem{corollary}[theorem]{Corollary}
\newtheorem{remark}[theorem]{Remark}

\def\ind{{\bf 1}\hskip-2.5pt{\rm l}}

\begin{document}

\title{Consistency of the maximum likelihood estimate for
Non-homogeneous Markov-switching models}

\author{Pierre Ailliot\\
Fran\c coise P\`ene\\
{\small \em Laboratoire de Math\'ematiques, UMR 6205, Universit\'e de Brest, France}\\}

\date{\today}

\maketitle

\begin{abstract}

Many nonlinear time series models have been proposed in the last decades. Among them, the models with regime switchings provide a class of versatile and interpretable models which have received a particular attention in the literature. In this paper, we consider a large family of such models which generalize the well known Markov-switching AutoRegressive (MS-AR) by allowing non-homogeneous switching and encompass Threshold AutoRegressive (TAR) models and prove the consistency of the maximum likelihood estimator under general conditions. We show that these conditions apply to specific but representative models with non-homogeneous Markov switchings. The famous MacKenzie River lynx dataset is used to illustrate one of these models.

\vspace{.5cm}

\end{abstract}

\section*{Introduction}

Recent decades have seen extensive interest in time series models with regime switchings. One of the most influential paper in this field is the one by Hamilton in 1989 (see \cite{Ham89}) where Markov-Switching AutoRegressive (MS-AR) models were introduced. It became one of the most popular nonlinear time series model. MS-AR models combine several autoregressive models to describe the evolution of the observed process $\{Y_k\}$ at different periods of time, the transition between these autoregressive models being controlled by a hidden Markov chain $\{X_k\}$. In most applications, it is assumed that $\{X_k\}$ is an homogeneous Markov chain. In this work, we relax this assumption and let the evolution of $\{X_k\}$ depend on lagged values of $\{Y_k\}$ and exogenous covariates. 

More formally, we assume that $X_k$ takes its values in a compact metric space 
$E$ endowed with a finite Borel measure $\mathfrak m_E$ and that  
$Y_k$ takes its values in a complete separable metric space 
$K$ endowed with a non-negative Borel $\sigma$-finite measure $\mathfrak m_K$ and we set $\mu_0:=\mathfrak m_E\times\mathfrak  m_K$. It will be useful to denote $Y_k^{k+\ell}:=(Y_k,...,Y_{k+\ell})$, $y_{k}^{k+\ell}:=(y_k,...,y_{k+\ell})$ 
(and to use analogous notations $X_k^{k+\ell}$, $x_k^{k+\ell}$) for integer $k$ and $\ell\ge 0$. 
The Non-Homogeneous Markov-Switching AutoRegressive (NHMS-AR) model of order $s>0$ considered in this work is characterized by Hypothesis \ref{hyp1} below.

\begin{hypo}\label{hyp1}
The sequence $\{X_k,Y_k\}_k$ is a Markov process of order $s$ with values in $E\times K$ such that, for some parameter $\theta$ belonging to
some subset $\Theta$ of $\mathbb R^d$,
\begin{itemize}
\item  the conditional distribution of $X_k$ (wrt $\mathfrak m_E$) 
given the values of $\{X_{k'}=x_{k'}\}_{k'<k}$ and $\{Y_{k'}=y_{k'}\}_{k' < k}$ only depends on $x_{k-1}$ and $y_{k-s}^{k-1}$ and 
this conditional distribution has a probability density function (pdf) denoted
$$p_{1,\theta}(x_k|x_{k-1},y_{k-s}^{k-1})$$
with respect to $\mathfrak m_E$.
\item the conditional distribution of $Y_k$ given the values of $\{Y_{k'}=y_{k'}\}_{k'<k}$ and 
$\{X_{k'}=x_{k'}\}_{k' \leq k}$ only depends on $x_{k}$ and $y_{k-s}^{k-1}$ and 
this conditional distribution has a pdf
$$p_{2,\theta}\left(y_k|x_k,y_{k-s}^{k-1}\right)$$
with respect to $\mathfrak m_K$.
\end{itemize}
\end{hypo}

Let us write $q_\theta(\cdot|x_{k-1},y_{k-s}^{k-1})$ for the conditional pdf (with respect to $\mu_0$)
of $(X_k,Y_k)$ given $(X_{k-1}=x_{k-1},Y_{k-s}^{k-1}=y_{k-s}^{k-1})$. 
Hypothesis \ref{hyp1} implies that 
\begin{equation}
\nonumber
q_\theta(x,y|x_{k-1},y_{k-s}^{k-1})=p_{1,\theta}(x|x_{k-1},y_{k-s}^{k-1})p_{2,\theta}(y|x,y_{k-s}^{k-1}) .
\end{equation}

The various conditional independence assumptions of Hypothesis \ref{hyp1} are summarized by the directed acyclic graph (DAG) below when $s=1$.
	\begin{equation}
				\nonumber
				\begin{array}[t]{lccccccccc}
 	\text{\textbf{Hidden Regime}}\ \ 		& \cdots & \rightarrow & X_{k-1} &  \rightarrow & 	X_{k} & \rightarrow & X_{k+1} & \rightarrow & \cdots \\
 				& & & \downarrow & \nearrow& \downarrow &\nearrow & \downarrow  & & \\
\text{\textbf{Observed time series}}\ \  & \cdots & \rightarrow & Y_{k-1} & \rightarrow & Y_{k} & \rightarrow & Y_{k+1} & \rightarrow & \cdots \\
				\end{array}
				\end{equation}

This defines a general family of models which encompasses the most usual models with regime switchings.
\begin{itemize}
	\item When $p_{1,\theta}(x_k|x_{k-1},y_{k-s}^{k-1})$ does not dependent on $y_{k-s}^{k-1}$, the evolution of the hidden Markov chain $\{X_k\}$ is homogeneous and independent of the observed process and we retrieve the usual MS-AR models. If we further assume that $p_{2,\theta}\left(y_k|x_k,y_{k-s}^{k-1}\right)$ does not depend of $y_{k-s}^{k-1}$, we obtain the Hidden Markov Models (HMMs).  		
\item  When $p_{1,\theta}(x_k|x_{k-1},y_{k-s}^{k-1})$ does not dependent on $x_{k-1}$ and is parametrized using  indicator functions, we obtain the Threshold AutoRegressive (TAR) models which is an other important family of models with regime switching in the literature (see e.g. \cite{Tong90}). 
\end{itemize}
HMMs, MS-AR and TAR models have been used in many fields of applications and their theoretical properties have been extensively studied (see e.g. \cite{Tong90}, \cite{Fan03} and  \cite{Cappe05}). 

Models with non-homogeneous Markov switchings have also been considered in the literature. In particular, they have been used to describe breaks  associated with events such as financial crises  or abrupt changes in government policy  in econometric time series (see \cite{Kim08} and references therein). They are also popular for  meteorological applications (see e.g. \cite{Hughes99}, \cite{Bellone00b}, \cite{Vrac07}, \cite{Ail12}) with the regimes describing the so-called "weather types". The most usual method procedure to fit such models consists in computing the Maximum Likelihood Estimates (MLE). It is indeed relatively straightforward to adapt the standard numerical estimation which are available for the homogeneous models, such as the forward-backward recursions or the EM algorithm, to the non-homogeneous models (see e.g. \cite{Dieb94}, \cite{Kim08}, \cite{Hughes99}). However, we could not find any theoretical results on the asymptotic properties of the MLE for these models and this paper aims at filling this gap.


The paper is organized as follows. In Section \ref{sec:theory}, we give general conditions ensuring the consistency of the MLE. They include conditions on the ergodicity of the model and the identifiability of the parameters. In Sections \ref{sec:appli} and \ref{sec:real}, we show that these general conditions apply to various specific but representative NHMS-AR models. Some results are proven in the appendices.
\section{A general consistency result of MLE for NHMS-AR models}
\label{sec:theory}
We aim at estimating the true parameter $\theta^*\in\Theta$ 
of a NHMS-AR process $(X_k,Y_k)_k$ for which only the component $\{Y_k\}$ is observed. 
For that we consider the Maximum Likelihood Estimator (MLE) 
$\hat \theta_{n,x_0}$ which is defined as the maximizer of
 $\theta\mapsto \ell_n(\theta,x_0)$ for a fixed $x_0\in E$ with
\begin{equation}
\nonumber
\ell_n(\theta,x_0)=\log  p_\theta(Y_1^n|X_0=x_0,Y_{-s+1}^0)
    =\sum_{k=1}^n\log\frac{p_\theta(Y_1^k|X_0=x_0,Y_{-s+1}^0)}
            {p_\theta(Y_1^{k-1}|X_0=x_0,Y_{-s+1}^0)},
\end{equation}
where $p_\theta(Y_1^k|X_0=x_0,Y_{-s+1}^0)$ is the conditional pdf of $Y_1^k$
given $(X_0=x_0,Y_{-s+1}^0)$ evaluated at $Y_1^k$, i.e.
$$p_\theta(Y_1^{k}|X_0=x_0,Y_{-s+1}^0):= 
     \int_{E^k}\prod_{\ell=1}^k q_\theta(x_\ell,Y_\ell|x_{\ell-1},
       Y_{\ell-s}^{\ell-1})\, d\mathfrak m_E^{\otimes k}(x_1^k).$$
Observe that $\ell_n(\theta,x_0)$ is a random variable depending on $Y_{-s+1}^0$ (which is observed).
\\
Before stating our main result, we introduce quickly some notations (see beginning of Appendix \ref{append2} for further details).
Let $Q_\theta$ be the transition operator of the $s$-order Markov process $(X_k,Y_k)_k$, $Q_\theta$ being seen as an operator acting on 
the set of complex-valued bounded measurable functions on $E\times K^s$ (or on some other complex Banach space)
and let $Q_\theta^*$ be its adjoint operator. We set $\mu:=\mathfrak m_E\times\mathfrak  m_K^{\otimes s}$.
We identify $(X_k,Y_k)_{k}$ with the canonical Markov chain.  We suppose that, for every $\theta\in\Theta$, there exists a
stationary probability $\bar\nu_\theta$ for the Markov chain with transition operator $Q_\theta$ (i.e. $\bar\nu_\theta$ is an invariant probability measure  for $Q_\theta^*$) with pdf $h_\theta$ with respect to $\mu$.
We then write $\bar {\mathbb P}_{\theta}$ for the probability measure corresponding to this invariant probability.
For every $P\in E\times K^s$ and any integer $k\ge s$, we write $Q^{*k}(\cdot|P)$ for the pdf of $(X_k,Y_k^{k+s-1})$ with
respect to $\mu$ given $(X_0,Y_0^{s-1})=P$.

The question of consistency of the MLE has been studied by many authors
in the context of usual HMMs (see e.g. \cite{Leroux92,LeGland-Mevel00,Douc-Matias01}) and MS-AR models (see \cite{DMR} and references therein). The aim of this section is to state consistency results of MLE for general NHMS-AR. 
The proof of the following theorem is a careful adaptation of the proof of \cite[Thm. 1 \& 5]
{DMR}. This proof is given in appendix \ref{append2}.

\begin{theorem}
\label{consistancegene}
Assume that $\Theta$ is compact, that is $(X_k,Y_k)_{k}$ ergodic,
that there exists an invariant probability measure for every $\theta\in\Theta$,
that  $\bar {\mathbb P}_{\theta^*}$ is absolutely continuous with respect to $\bar {\mathbb P}_\theta$ for every 
$\theta\in\Theta$, that $p_1$ and $p_2$ are continuous in $\theta$.
Assume also that the following conditions hold true
\begin{equation}\label{HH1}
0<p_{1,-}:=\inf_{\theta,x_1,x_{0},y_0} p_{1,\theta}(x_1|x_0,y_0)\le p_{1,+}
:=\sup_{\theta,x_1,x_{0},y_0} p_{1,\theta}(x_1|x_0,y_0)<\infty,
\end{equation}
\begin{equation}\label{HH3}
B_-:=\bar{\mathbb E}_{\theta^*} \left[\left|\log\left(\inf_\theta\int_E p_{2,\theta}
        (Y_0|x_0,Y_{-s}^{-1})\, d\mathfrak m_E(x_0)\right)\right|\right]<\infty,
\end{equation}
\begin{equation}\label{HH3b}
B_+:=\bar{\mathbb E}_{\theta^*} \left[\left|\log\left(\sup_\theta\int_E p_{2,\theta}
        (Y_0|x_0,Y_{-s}^{-1})\, d\mathfrak m_E(x_0)\right)\right|\right]<\infty,
\end{equation}
\begin{equation}\label{HH3c}
\forall\theta\in\Theta,\ \ 
   \sup_{y_{-s}^{-1}}\int_Ep_{2,\theta}(Y_0|x,y_{-s}^{-1})\, d\mathfrak m_E(x)<\infty,\ \ 
     \bar{\mathbb  P}_{\theta^*}-a.s,
\end{equation}
\begin{equation}\label{HH4}
\forall\theta\in\Theta,\ \ for\ \mu-a.e.\ P\in E\times K^s,  
 \lim_{k\rightarrow+\infty}|| Q_\theta^{*k}(\cdot|P)-h_\theta||_{L^1(\mu)}=0.
\end{equation}
Then, for every $x_0\in E$, the limit values of $(\hat\theta_{n,x_0})_n$ 
are $\bar{\mathbb P}_{\theta^*}$-almost surely
contained in $\{\theta\in\Theta\ :\ \bar{\mathbb P}^{Y}_\theta=\bar{\mathbb P}^Y_{\theta^*}\}$. 

If, moreover, $Q_{\theta^*}$ is positive Harris recurrent and aperiodic, then,
for every $x_0\in E$ and every initial probability $\nu$, the limit values of $(\hat\theta_{n,x_0})_n$ 
are almost surely
contained in $\{\theta\in\Theta\ :\ \bar{\mathbb P}^Y_\theta=\bar{\mathbb P}^Y_{\theta^*}\}$.
\end{theorem}

Our hypotheses are close to
those of \cite{DMR}. Let us point out the main differences.
First, in \cite{DMR} $p_{1,\theta}(x|x',y')$ does not depend on $y'$.
Second, (\ref{HH3b}) and (\ref{HH3c}) are slightly weaker 
than $$\sup_{\theta, y_{-s}^{-1},y_0,x} p_{2,\theta}(y_0|x,y_{-s}^{-1})<\infty$$ assumed in \cite{DMR}.  This is illustrated below in Section \ref{sec:real} where the parametrization of $p_2$ uses Gamma pdf which may not be bounded close 
to the origin depending on the values of the parameters. The results given in 
\cite{DMR} do not apply directly to this model whereas we will show that (\ref{HH3b}) applies (see also \cite{Ail06}).
Third, to prove the result in the stationary case, we replace Harris recurrence by (\ref{HH4}) which is equivalent
to each one of the two following properties 
\begin{itemize}
\item for any initial measure $\nu$ on $E\times K^s$, we have
$
\lim_{n\rightarrow +\infty}\left\Vert  Q_\theta^{*n}\nu-\nu_\theta\right\Vert_{TV}=0$,
where $\Vert\cdot\Vert_{TV}$ stands for the total variation norm,
\item for any initial measure $\nu$ on $E\times K^s$, we have
$\lim_{n\rightarrow +\infty}\sup_{\nu\in\mathcal P(E\times K)}||[Q_\theta^{*n}\nu]-h_\theta||_{L^1(\mathfrak m_E\times \mathfrak m_K^s)}=0$,
with $\mathcal P(E\times K)$ the set of probability measures
on $E\times K$.
\end{itemize}
\begin{remark}
Observe that, if $q_\theta>0$ and if $\nu_\theta$ exists for every
$\theta\in\Theta$, then the pdf $h_\theta$ of $\nu_\theta$ satisfies
$h_\theta>0$ ($\mu$-a.e.). In this case, $\bar {\mathbb P}_{\theta^*}$ is absolutely continuous with respect to $\bar {\mathbb P}_\theta$ for every 
$\theta\in\Theta$.

Observe also that the ergodicity of the dynamical system 
$(\Omega,\mathcal F,\bar{\mathbb P}_{\theta^*},\tau)$ is satisfied as soon as
the transition operator is strongly ergodic with respect some Banach space
$\mathcal B$ satisfying general assumptions (see for example \cite[Proposition 2.2]{Loicsoaz1}).
\end{remark}
\section{NHMS-AR model with linear autoregressive models}\label{sec:appli}

\subsection{A NHMS-AR model for MacKenzie River lynx data}
\label{sec:lynx}

In this section we introduce a particular NHMS-AR model and discuss the results obtained when fitting this model to the the time series of annual number of Canadian lynx trapped in the Mackenzie River district of northwest Canada from 1821 to 1934. This time series is a benchmark dataset to test nonlinear time series model (see e.g. \cite{Tong90}, \cite{Fan03}). In order to facilitate the comparison with the other works on this time series, we analyze the data at the logarithm scale with the base 10 shown on Figure \ref{fig:datalynx}. This time series exhibits periodic fluctuations (it may be due to the competition between several species, predator-prey interaction,...) with asymmetric cycles (increasing phase are slower than decreasing phase) which makes it challenging to model.

\begin{figure}[!ht]
\centering
\makebox{\includegraphics[scale=.75]{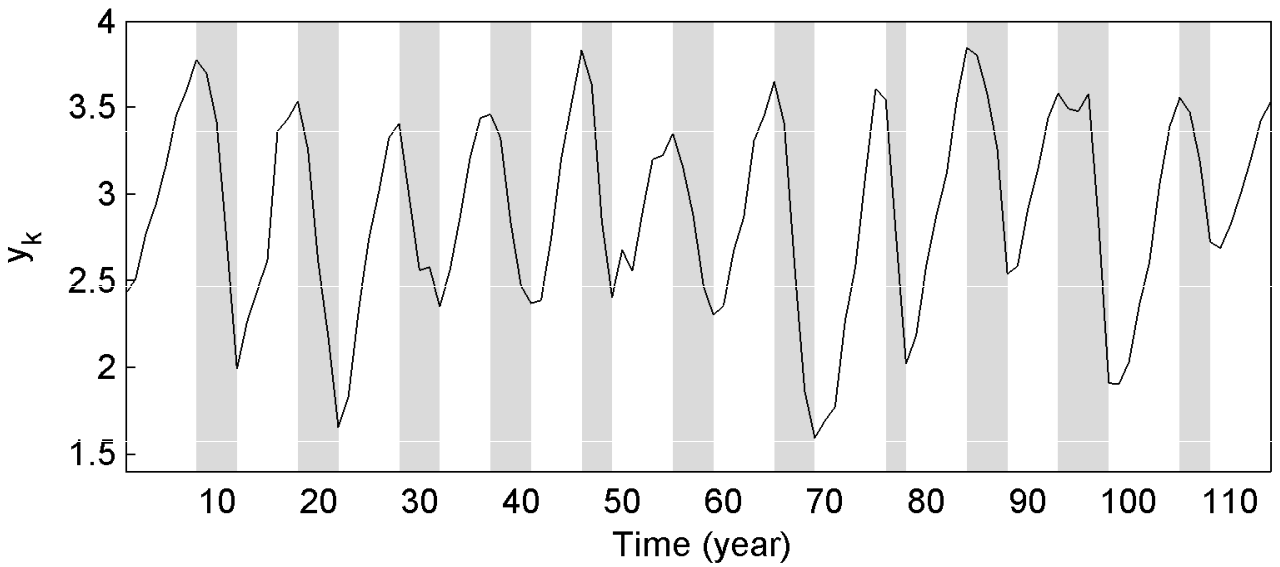} \includegraphics[scale=.75]{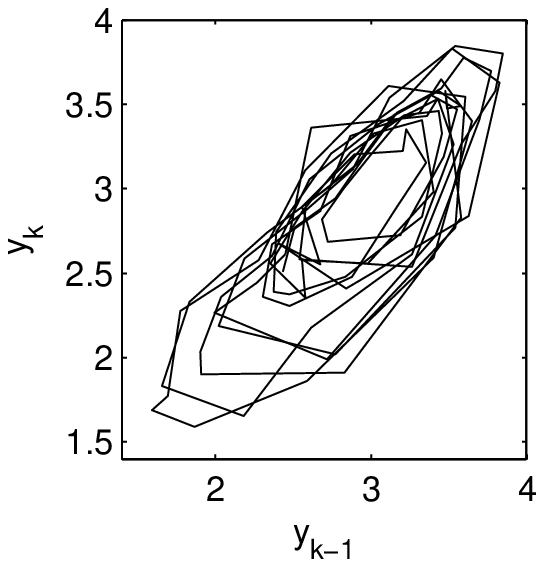} }\\
\makebox{\includegraphics[scale=.75]{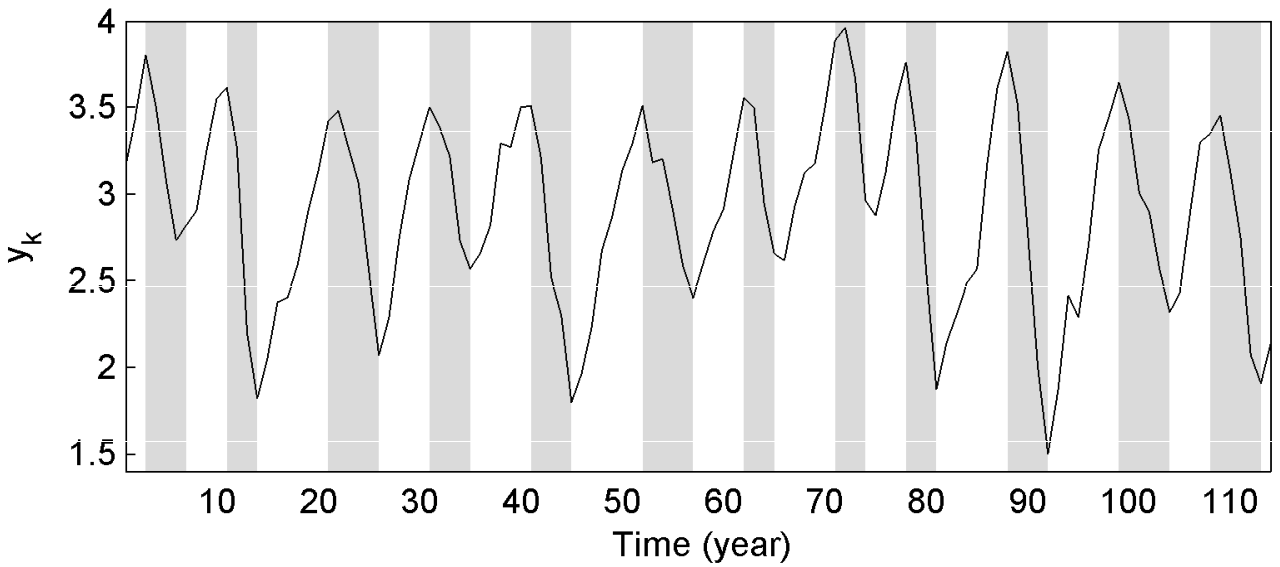} \includegraphics[scale=.75]{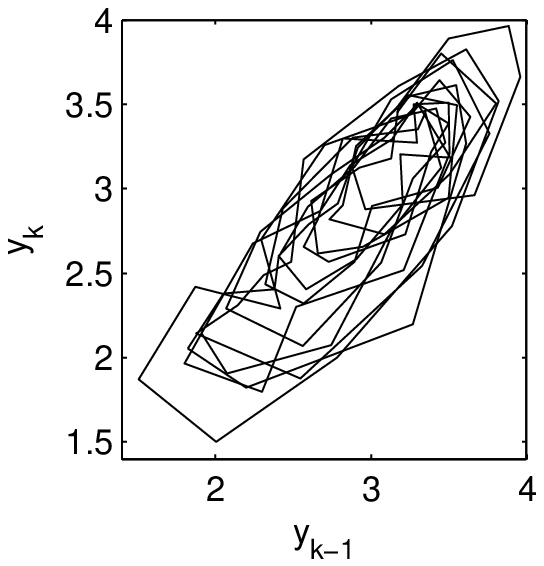}}
\caption{\label{fig:datalynx} {Top left panel: time plot of log Canadian lynx data. The color indicates the most likely regimes identified by the fitted NHMS-AR model. The first [resp. second] regime is the most likely when the color is white [resp. gray]. Top right panel: directed scatter plot of log Canadian lynx data. Bottom left panel: time plot of a sequence simulated with the fitted NHMS-AR model data. The color indicates the simulated regime (first regime in white, second regime in gray). Bottom right panel: directed scatter plot of the simulated sequence shown on the bottom left panel.}}
\end{figure}

In \cite{Tong90}, it was proposed to fit a SETAR(2) model to this time series. The fitted model is the following
\begin{equation}
\label{eq:TARlynx}
Y_k=\left\{\begin{array}{ll}
0.51 + 1.23 Y_{k-1}  - 0.37 Y_{k-2} +0.18 \epsilon_k & (Y_{k-2} \leq 3.15) \\
2.32 + 1.53 Y_{k-1}  - 1.27 Y_{k-2} +0.23\epsilon_k & (Y_{k-2} > 3.15) 
\end{array}\right. .
\end{equation}
The two regimes have a nice biological interpretation in terms of prey-predator interaction, with the upper regime ($Y_{t-2} > 3.15$) corresponding to a population decrease whereas the population tends to increase in the lower regime. 

The NHMS-AR model defined below has been fitted to this time series.
\begin{hypo}\label{hypAR}
We assume that $E=\{1,2\}$ (endowed with the counting measure), 
$K=\mathbb R$ (endowed with the Lebesgue measure)  and $\{Y_k\}$ satisfies
\begin{equation}
\label{eq:ARgauss}
\nonumber
Y_k=\beta_0^{(x_k)}+\sum_{\ell=1}^s \beta_{\ell}^{(x_k)} Y_{k-\ell} + \sigma^{(x_k)} \epsilon_k
\end{equation}
with $\{\epsilon_k\}$ an iid sequence of standard Gaussian random
variables, with $\sigma^{(x)}>0$ and $\beta_l^{(x)} \in \mathbb R$ for
every $\ell \in \{0,...,s\}$ and every $x \in \{1,2\},$
\begin{equation}\label{eq:AR1}
i.e.\ \ \ \ \ 
p_{2,\theta}(y_k|x_k,y_{k-s}^{k-1})=\mathcal N\left(y_k;\beta_0^{(x_k)}+\sum_{\ell=1}^s\beta_\ell^{(x_k)}y_{k-\ell},\sigma^{(x_k)}\right),
\end{equation}
where $\mathcal N(\cdot;\beta,\sigma)$ stands for the Gaussian pdf
with mean $\beta$ and standard deviation $\sigma$.

The transition probabilities of $\{X_k\}$ are parametrized using the logistic function as follows when $x_k=x_{k-1}$
\begin{equation}
\label{eq:nhlog}
p_{1,\theta}(x_k|x_{k-1},y_{k-s}^{k-1})=\pi_{-}^{(x_{k-1})} +\frac{1-\pi_{-}^{(x_{k-1})}-\pi_{+}^{(x_{k-1})}}{1+\exp\left(\lambda_0^{(x_{k-1})}+\lambda_1^{(x_{k-1})} y_{k-r}\right)} 
\end{equation}
with $r \leq s$ a positive integer and the unknown parameters $\pi_{-}^{(x)},\pi_{+}^{(x)},\lambda_0^{(x)}, \lambda_1^{(x)}$ for $x \in \{1,2\}$. \\
The unknown parameter $\theta$ corresponds to
$$\theta=\left((\beta_i^{(x)} ),(\sigma^{(x)}),(\pi_{-}^{(x)}),(\pi_{+}^{(x)}),(\lambda_i^{(x)})\right).$$
We write $\tilde\Theta$ for the set of such parameters $\theta$ satisfying, for every
$x\in\{1,2\}$, 
$\sigma^{(x)}>0$ and $0<\pi_{-}^{(x)}<1-\pi_{+}^{(x)}<1$ 
(this last constraint is added in order to ensure that (\ref{HH1}) holds).
\end{hypo}

Although very simple, this model encompasses the homogeneous Gaussian
MS-AR model when $\lambda_1^{(1)}=\lambda_1^{(2)}=0$ and the SETAR(2) model as a limit case. Indeed,  if $s=-\frac{\lambda_0^{(x)}}{\lambda_1^{(x)}}$ is fixed for $x\in\{1,2\}$, $\lambda_1^{(1)}\rightarrow +\infty$,  $\lambda_1^{(2)}\rightarrow -\infty$, $\pi_{-}^{(x)}\rightarrow 0$ and  $\pi_{+}^{(x)}\rightarrow 0$ then
\begin{equation}
\nonumber
p_1(X_k=1|x_{k-1},y_{k-s}^{k-1}) \rightarrow \ind(y_{k-r}\leq s) 
\ 
\text{and}
\ 
p_1(X_k=2|x_{k-1},y_{k-s}^{k-1}) \rightarrow \ind(y_{k-r}\geq s)
	\end{equation}
Both models have been extensively studied in the literature.

 In practice, we have used the EM algorithm to compute the MLE. The recursions of this algorithm are relatively similar to the ones of the MS-AR model (see \cite{krolzig97}, \cite{Dieb94}). To facilitate the comparison with the SETAR(2) model (\ref{eq:TARlynx}), we have also considered AR models of order $s=2$ and a lag $r=2$ for the transition probabilities.  The fitted model is the following 

\begin{equation}
\label{eq:NHlynx}
Y_k=\left\{\begin{array}{llllllll}
0.54  &+ 1.11  &Y_{k-1}  &- 0.24  &Y_{k-2} &+0.14  &\epsilon_k & (X_k=1) \\
\textit{(0.31,0.80)} & \textit{(0.96,1.27)}&&\textit{(-0.43,-0.05)} &&\textit{(0.11,0.17)}& &\\
1.03   &+ 1.49 & Y_{k-1}  &- 0.87 &  Y_{k-2} &+0.22  & \epsilon_k & (X_k=2) \\
\textit{(-0.12,1.86)} &\textit{(1.23,1.69)}&&\textit{(-1.20,-0.39)}&&\textit{(0.14,0.26)}&&
\end{array}\right .
\end{equation}

 with
 \begin{equation}
\label{eq:translynx}
P\left(X_k=i|X_{k-1}=i,Y_{k-2}=y_{k-2}\right)=\left\{\begin{array}{lllll}
(1+exp(&-42.4&+12.8 &y_{k-2}))^{-1} & (X_k=1) \\
&\textit{(-587,-16.3)} &\textit{(4.77,176)} &&\\
(1+exp(&9.07&-3.33 &y_{k-2}))^{-1}  & (X_k=2)\\
&\textit{(2.25,178)} &\textit{(-64.1,-1.12)} &&
\end{array}\right .
\end{equation}

where the italic values in parenthesis below the parameter values correspond to 95\% confidence intervals computed using parametric bootstrap (see e.g. \cite{visser2000}). These values reflect the finite sample properties of the estimates. The estimate of $\pi_{-}^{(x)}$ and $\pi_{+}^{(x)}$ are not given because they are very close to $0$. It means that these technical parameters have no practical importance and can be fixed equal to an arbitrary small value (here we used the machine epsilon $2^{-52}$). 
There are small differences between the AR coefficients (\ref{eq:TARlynx}) and (\ref{eq:NHlynx}) but the dynamics inside the regimes of the SETAR(2) and NHMS-AR models are broadly similar.
 The models differ mainly in the mechanism used to govern the switchings between the two regimes. For the SETAR model the regime is a deterministic function of a lagged value of the observed process. The NHMS-AR model can be seen as a fuzzy extension of the SETAR model where the regime has its own Markovian evolution influenced by the lagged value of the observed process. This is illustrated on Figure~\ref{fig:translynx} which shows the transition probabilities (\ref{eq:translynx}) and the threshold of the SETAR(2) model. According to this figure, it seems reasonable to model the transition from regime 1 to regime 2 by a step function at the level $y_{k-2} \approx 3.15$ but the values of $y_{k-2}$ for which the transition from regime 2 to regime 1 occurs seem to be more variable and the step function approximation less realistic.

The asymmetries in the cycle imply that the system spends less time in the second regime (decreasing phase) than in the first one. It may explain the larger confidence intervals in the second regime compared to the first one (see  (\ref{eq:NHlynx})). Figure \ref{fig:translynx} shows that there is an important sampling variability in the estimate of the transition kernel of the hidden process. This is probably due to the low number of transitions among regimes (see Figure \ref{fig:datalynx}) which makes it difficult to estimate the associated parameters. A similar behavior has been observed when fitting the model to other time series.

\begin{figure}[!ht]
\centering
\makebox{\includegraphics[scale=.75]{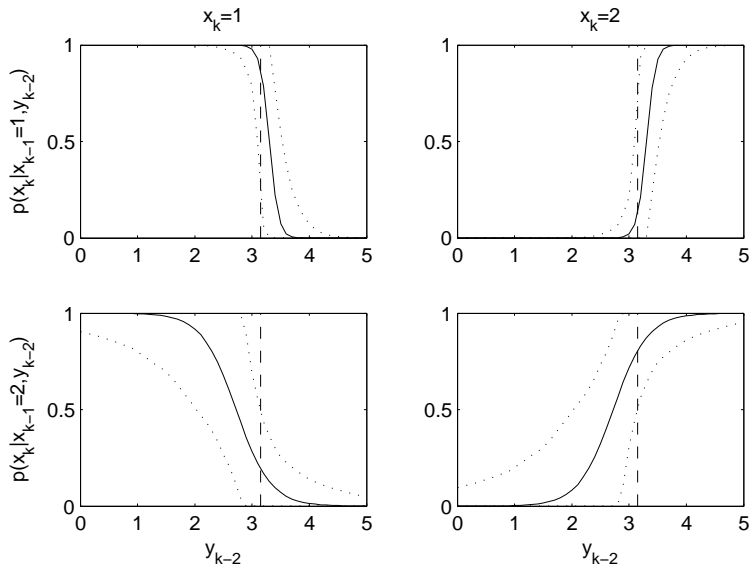}}
\caption{\label{fig:translynx} {Transition probabilities $P\left(X_k=j|X_{k-1}=i,Y_{k-2}=y_{k-2}\right)$ as a function of $y_{k-2}$. The dotted lines correspond to 95\% confidence intervals computed using parametric bootstrap. The dashed vertical line corresponds to the threshold (3.15) of the SETAR(2) model.}}
\end{figure}

Table \ref{tab:BIC} gives the AIC and BIC values defined as
$$AIC = -2 log L + 2npar, \ \ \  BIC = -2 log L + npar \log(N)$$
and $L$ is the likelihood of the data, $npar$ is the number of parameters and $N$ is the number of observations. The values for the NHMS-AR and SETAR models are relatively similar. The NHMS-AR models has a slightly better AIC value but BIC selects the SETAR model. As expected, these two models clearly outperform the homogeneous MS-AR which does not include information on the past values in the switching mechanism.

\begin{table}[!ht]
\centering
\begin{tabular}{|c|ccc|}
\hline
& AIC & BIC & npar \\
\hline 
SETAR ($s=2$) & -28.33 & -3.70 &9\\
MS-AR ($s=2$) & -0.2063 & 27.15 & 10\\
NHMS-AR ($r=s=2$)& -30.83 & 2.00 &12\\

\hline
\end{tabular}
\caption{AIC and BIC values for the fitted SETAR, homogeneous MS-AR  and NHMS-AR models }
\label{tab:BIC}
\end{table}

The simulated sequence shown on Figure \ref{fig:datalynx} exhibits a similar cyclical behavior than the data. A more systematic validation was performed but the results are hard to analyze because of the low amount of data available. The model can be generalized in several ways to handle $M\geq3$ regimes or include covariates, for example through a linear function in the logistic term (see  e.g. \cite{Dieb94}). Other link functions, such as the probit model used in \cite{Kim08} or a Gaussian kernel (see (\ref{eq:p1rain})), or non-linear autoregressive models could also be considered. Such models have been developed for various environmental data including temperature and wind time series. The fitted models generally provide an accurate description of the distributional properties of these time series and accurate short-term forecasts. This will be the subject of a forthcoming paper.


\subsection{Properties of the Markov chain}

In this section, we discuss the recurrent and ergodic properties of the model introduced in the previous section. It is a key step to prove the consistence of the MLE (see Theorem \ref{consistancegene}).
Various authors have studied the ergodicity  of MS-AR (\cite{Yao2000}, \cite{Yao2001}, \cite{Francq98}) and TAR (\cite{Chen91}, \cite{an96}) models. A classical approach to prove the ergodicity of a non-linear time series consists in establishing a drift condition. 
Here we will use a strict drift condition. Let $\Vert\cdot\Vert$ be some norm 
on $\mathbb R^s$. For any $R>0$, we consider
the set $E_R:=\{(x,y_{-s+1}^0):\Vert y_{-s+1}^0\Vert\le R\}$.
Recall that $\mu$ is here the product of the counting measure on $E$ and of the Lebesgue measure on $\mathbb R^s$.
\begin{prop}
Assume Hypothesis \ref{hypAR}.

The Markov chain is $\psi$-irreducible (with $\psi=\mu$).

Let $R>0$. The set $E_R$ is $\nu_s$-small and $\nu_{s+1}$-small
with $\nu_s$ and $\nu_{s+1}$ equivalent to $\mu$.
Hence, the markov chain is aperiodic.
\end{prop}
\begin{proof}
The $\psi$-irreducibility comes from the positivity of $q_\theta$.
Let us prove that $E_R$ is $\nu_s$-small with $\nu_s=h_s\cdot\mu$
and
$$h_s(x_s,y_{1}^s)=\inf_{(x_0,y_{-s+1}^{0})\in E_R}\int_{E^s}\prod_{\ell=1}^s
   q_\theta(x_\ell,y_\ell|x_{\ell-1},y_{\ell-s}^{\ell-1})\, dx_1^{s-1}>0. $$
Indeed $p_{1,\theta}$ is uniformly bounded from below by some $p_{1,-}$, $\sigma^{(x)}$
are uniformly bounded from above by some $\sigma_+$ and from below
by some $\sigma_-$ and, for every $\ell\in\{1,...,s\}$, we have
$$\forall Z\in\mathbb R,\ \ \
   g_\ell(Z):=\sup_{(x_\ell,y_{-s+1}^{0})\in E_R}\left|Z-\beta_0^{(x_\ell)}-\sum_{j=\ell}^s
     \beta_j^{(x_\ell)}y_{\ell-j}\right|^2<\infty. $$
So
$$h_s(x_s,y_1^s)\ge\inf_{x_1,...,x_s\in\{1,2\}}
   \frac{({p_{1,-}})^s}{(2\pi \sigma_-)^{\frac s2}}
   \exp\left(-\frac 1{2\sigma_+}\sum_{\ell=1}^sg_\ell\left(y_\ell
     - \sum_{j=1}^{\ell-1}
     \beta_j^{(x_\ell)}y_{\ell-j}\right)\right) .$$
The proof of the $\nu_{s+1}$-smallness of $E_R$ (with $\nu_{s+1}$ equivalent to $\mu$) uses the same ideas.
\end{proof}
Now, to obtain the other properties related to the ergodicity of the process for practical applications
(including the practical example given in Section \ref{sec:lynx}), we can use the following
strict drift property.
\begin{hypo}\label{strictdrift}
There exist three real numbers $K<1$, $L>0$ and $R>0$ such that, for
every $(x_0,y_{-s+1}^0)\in\{1,2\}\times\mathbb R^s$,
\begin{equation}
\label{eq:drift}
\mathbb E[\left\|Y_{-s+2}^1\right\|^2|Y_{-s+1}^0=y_{-s+1}^0,X_0=x_0] \leq K \left\|y_{-s+1}^0\right\|^2 +L \ind_{E_R}(y_{-s+1}^0).
\end{equation}
\end{hypo}
Recall that this property has several classical consequences 
(see \cite[Chapters 11 and 15]{MeynTweedie} for more details). Hypothesis \ref{strictdrift}
(combined with the irreducibility and aperiodicity coming from Hypothesis \ref{hypAR}) implies in particular
\begin{itemize}
\item the existence of a (unique) stationary measure admitting a moment of
order 2;
\item the $V$-geometric ergodicity with $V(x,y_{-s+1}^0)=\Vert
   y_{-s+1}^0\Vert^2$ and so the ergodicity of the Markov chain (see
for example \cite[Proposition 2.2]{Loicsoaz1} for this last point);
\item the positive Harris recurrence.
\end{itemize}
We end this section with some comments on (\ref{eq:drift}).
Let us write
$$\Lambda^{(x)} = \left(\begin{array}{cccccc} 
 0 & 1 & 0 & 0 &\cdots& 0\\
  0 & 0 & 1 & 0 &\cdots & 0\\\
 \vdots&\vdots&\vdots&\ddots&\vdots\\
 0 & 0 & 0 & 0 &\cdots & 1\\
 \beta_{s}^{(x)}&\beta_{s-1}^{(x)} &\cdots&\cdots&\cdots&\beta_{1}^{(x)}
 \end{array}
 \right)$$
for the companion matrix associated to the AR model in regime $x$, 
$$\Phi^{(x)}:= \left(\begin{array}{c} 
 0\\
 0\\
 \vdots\\
 0\\
 \beta_{0}^{(x)}
 \end{array}
 \right)
 ,\ \ 
\Sigma^{(x)}=\left(\begin{array}{cccc}
 0 & \cdots &  0 & 0\\
 \vdots&\vdots&\vdots&\vdots\\
0 & \cdots &  0 & 0\\
0 & \cdots &  0 & \sigma^{(x)}
  \end{array}
 \right)\ \ \mbox{and}\ \ \boldsymbol{\varepsilon}:= \left(\begin{array}{c} 
 0\\
 0\\
 \vdots\\
 0\\
 \varepsilon_1
 \end{array}
 \right).$$
There exist $A,B>0$ such that, for every $(x_0,y_{-s+1}^0)\in\{1,2\}\times\mathbb R^s$, we have
\begin{eqnarray}
\nonumber
\mathbb E[\left\|Y_{-s+2}^1\right\|^2|Y_{-s+1}^0=y_{-s+1}^0,X_0=x_0] &=& \sum_{x_1=1}^M  p_{1,\theta}(x_1|x_0,y_{-s+1}^0) \mathbb E[\left\|\Lambda^{(x_1)} y_{-s+1}^0+\Phi^{(x_1)}+\Sigma^{(x_1)} \boldsymbol{\varepsilon}\right\|^2] \\
\nonumber
&\leq & \sum_{x_1=1}^M  p_{1,\theta}(x_1|x_0,y_{-s+1}^0) \left\|\Lambda^{(x_1)}\right\|^2  \left\|y_{-s+1}^0\right\|^2+A\left\|y_{-s+1}^0\right\|+B
\end{eqnarray}
where $\left\|.\right\|$ denotes abusively the matrix norm associated to the vector norm. We deduce the following.
\begin{remark}
The strict drift condition (\ref{eq:drift}) is satisfied when there exists $M>0$  such that for all $x_0 \in E$ and all $y_{-s+1}^0 \in {\mathbb R}^s$
\begin{equation}
\label{eq:driftb}
\Vert y_{-s+1}^0\Vert>R\ \ \Rightarrow\ \ \sum_{x_1 \in E}  p_{1,\theta}(x_1|x_0,y_{-s+1}^0)\left\| \Lambda^{(x_1)}\right\|^2 <1.
\end{equation}
This is true in particular if 
\begin{equation}
\label{eq:driftc} 
\forall x\in E, \; \left\|\Lambda^{(x)}\right\|<1.
\end{equation}
\end{remark}

The model fitted to the lynx data in the previous section satisfies condition (\ref{eq:driftc}) for the matrix norm defined as 
$$\left\| A \right\| = \left\| P^{-1}AP \right\|_\infty$$
with $P$ the matrix containing the eigenvectors of the companion matrix for the second regime and $\left\|.\right\|_\infty$ the infinity norm. This condition implies that all the regimes are stable.  However, it is also possible to construct models which satisfy (\ref{eq:driftb}) with some unstable regimes if the instability is controlled by the dynamics of $\{X_k\}$.

\begin{remark}
The results given in this section are still valid when
\begin{itemize}
	\item  the noise $\{\epsilon\}_k$ in (\ref{eq:ARgauss}) is an iid sequence with finite variance which admits a pdf $f$ with respect to the Lebesgue measure such that for all $R>0$, $\inf_{y \in E_R} f(y)> 0 $, and
	\item $E=\{1,...,M\}$ with $M \geq 2$ and (\ref{eq:nhlog}) replaced by any transition kernel $p_{1,\theta}$ satisfying (\ref{HH1}).
\end{itemize}
\end{remark}

\subsection{Consistency of MLE}
\label{MLEgaussien}
The results given in this section generalize the results given in \cite{Francq1998, krish1998} for homogeneous MS-AR models with linear Gaussian autoregressive models. 

\begin{corollary}\label{consistenceAR}
Assume that Hypotheses \ref{hypAR} and \ref{strictdrift} 
hold true for every $\theta$. Let $\Theta$ be a compact subset of $\tilde\Theta$. Then,
for all $\theta\in\Theta$ there exists a unique invariant probability
distribution and, for every $x_0\in M$ and every initial probability
distribution $\nu$, 
the limit values of $(\hat\theta_{n,x_0})_n$ are $\bar{\mathbb P}_{\theta^*}$-almost surely
contained in $\{\theta\in\Theta\ :\ \bar{\mathbb P}_\theta=\bar{\mathbb P}_{\theta^*}\}$. 
\end{corollary}
\begin{proof}
This corollary is a direct consequence of Theorem \ref{consistancegene} 
and of the previous section.
As already noticed in section \ref{sec:theory}, the invariant measure has a positive pdf with respect to $\mu$. As seen in the previous section, the Markov chain is aperiodic positive Harris recurrent (which implies (\ref{HH4})) and
the stationary process is square integrable, which implies (\ref{HH3}) and (\ref{HH3b}). In this example, $p_{2,\theta}$ is bounded from above and
so (\ref{HH3c}) holds.
\end{proof} 
\begin{remark}
Corollary \ref{consistenceAR} is still valid when $E=\{1,...,M\}$ with $M \geq 2$ and (\ref{eq:nhlog}) replaced by any transition kernel $p_{1,\theta}$ satisfying (\ref{HH1}).
\end{remark}
In the sequel, we explicit the limit set $\{\theta\in\Theta\ :\ \bar{\mathbb P}_\theta=\bar{\mathbb P}_{\theta^*}\}$ under the supplementary condition 
\begin{equation}\label{differentAR}
\left(\beta_0^{(1)},\beta_{1}^{(1)},... ,\beta_{s}^{(1)}, \sigma^{(1)} \right) \ne \left(\beta_0^{(2)},\beta_{1}^{(2)},... ,\beta_{s}^{(2)}, \sigma^{(2)} \right)
\end{equation}
 that the dynamics in the two regimes are distinct. Note that this condition is not sufficient in order to ensure identifiability. First, it can be easily seen that the homogeneous MS-AR model can be written in many different ways using the parametrization (\ref{eq:nhlog}). It led us to add one of the following constraints on the parameters 
 \begin{equation}
 \label{lambdadiff0}
 \forall x \in \{1,2\}, \lambda_1^{(x)} \ne 0
\end{equation}
 which does not include the homogeneous model as a particular case or
  \begin{equation}
 \label{pifix}
\forall x \in \{1,2\}, \pi_-^{(x)} = \pi_+^{(x)} = \pi_0 \ \text{where} \  0<\pi_0<1/2 \ \text{is a fixed constant}
\end{equation}
 in order to solve this problem. A practical motivation for (\ref{pifix}) is given in Section \ref{sec:lynx}.  Let $\Theta'$ be the set of $\theta\in\tilde\Theta$ satisfying (\ref{lambdadiff0}) and let
$\Theta''$ be the set of $\theta\in\tilde\Theta$ satisfying (\ref{pifix}).
 Then, a permutation of the two states also leads different parameters values but to the same model. This problem can be solved  by ordering the regimes or by allowing a permutation of the states as discussed below.

\begin{prop}[Identifiability]\label{identifiabilityAR}
Let $\theta_1$ and $\theta_2$ belong to $\Theta'$ (resp. $\Theta''$) with $\theta_i=\left(\theta_i^{(1)},\theta_i^{(2)}\right)$ and 

$$\theta_i^{(x)}=\left((\beta_{j,(i)}^{(x)})_{j \in \{0,...s\}},\sigma_{i},(\lambda_{j,(i)}^{(x)})_{j \in \{0,1\}}\right)$$
the parameters associated with the regime $x \in \{1,2\}$.

Assume that $\theta_1$ satisfies (\ref{differentAR}). Then
$\bar {\mathbb P}_{\theta_1}^Y=\bar {\mathbb P}_{\theta_2}^Y$ if and only if $\theta_1$ and $\theta_2$ define the same model up to a permutation of indices, i.e. there exists
a permutation $\tau$ of $\{1,2\}$ such that 
$$\theta_1^{(x)}=\theta_2^{(\tau(x))}$$
 
\end{prop}
The proof of Proposition \ref{identifiabilityAR} is postponed to appendix \ref{appendAR}. 

Now due to Corollary \ref{consistenceAR} and Proposition \ref{identifiabilityAR}, we directly get Theorem 
\ref{thmAR}.

\begin{theorem}\label{thmAR}
Assume that Hypotheses \ref{hypAR} and \ref{strictdrift} hold true
for every $\theta$. 
Let $\Theta$ be a compact subset of $\Theta'$ or $\Theta''$. Assume that $\theta^*$ satisfies (\ref{differentAR}).
Then, for every $x_0\in\{1,2\}$ and any initial probability distribution $\nu$, on a set of probability one, 
the limit values $\theta$ of the sequence of random variables 
$(\hat \theta_{n,x_0})_n$ are equal to $\theta^*$ 
up to a permutation of indices.
\end{theorem}



\section{Non-homogeneous Hidden Markov Models with exogenous variables}
\label{sec:real}

\subsection{Model}

When using NHMS-AR models in practice, it is often assumed that the evolution of $\{X_{k}\}$ depends not only on lagged values of the process of interest but also on strictly exogenous variables. In order to handle such situation, we will denote $Y_k=(Z_k,R_k)$ with $\{Z_k\}$ the time series of covariates and $\{R_k\}$ the output time series to be modeled. Besides Hypothesis \ref{hyp1}, various supplementary conditional independence assumptions can be made for specific applications. For example, in \cite{Hughes99} it is assumed that the switching probabilities of $\{X_k\}$ only depend on the exogenous covariates
\begin{equation}
\nonumber
p_{1,\theta}(x_k|x_{k-1},r_{k-s}^{k-1},z_{k-s}^{k-1})=p_{1,\theta}(x_k|x_{k-1},z_{k-1})
\end{equation}
that the evolution of $\{Z_k\}$ is independent of $\{X_k\}$ and $\{R_k\}$ and that $R_k$ is conditionally independent of $Z_{k-s}^{k}$ and $R_{k-s}^{k-1}$ given $X_k$
\begin{equation}
\label{eq:covariate}
\nonumber
p_{2,\theta}\left(z_k,r_k|x_k,z_{k-s}^{k-1},r_{k-s}^{k-1}\right)= p_{R,\theta}\left(r_k|x_k\right) p_{Z}\left(z_k|z_{k-1}\right).
\end{equation}

This model, which dependence structure is summarized by the DAG below when $s=1$ is often referred as Non-Homogeneous Hidden Markov Models (NHMMs) in the literature.
	\begin{equation}
				\nonumber
				\begin{array}[t]{lccccccccc}
			\text{\textbf{Covariates}}	& \cdots & \rightarrow & Z_{k-1} &  \rightarrow & 	Z_{k} & \rightarrow & Z_{k+1} & \rightarrow & \cdots \\
 				& & & &\searrow & & \searrow & & \searrow   & \\
 	\text{\textbf{Hidden Regime}}\ \ 					& \cdots & \rightarrow & X_{k-1} &  \rightarrow & 	X_{k} & \rightarrow & X_{k+1} & \rightarrow & \cdots \\
 & & & \downarrow & & \downarrow && \downarrow  & & \\
\text{\textbf{Output time series}}\ \  & \cdots &  & R_{k-1} & & R_{k} & & R_{k+1} &  & \cdots \\
				\end{array}
				\end{equation}			

In this section, we consider a typical example of NHMM with finite hidden state space and strictly exogenous variables and show that the theoretical results  proven in this paper apply to this model. We focus on a model initially introduced in \cite{Bellone00b} 
for downscaling rainfall. It is an extension of the model proposed in \cite{Hughes99} (see also \cite{Vrac07} for more recent references). The results given in this section can be adapted to other NHMM with finite hidden state space such as the one proposed in \cite{Dieb94} which is widely used in econometrics. The model is described more precisely hereafter.

\begin{hypo}\label{hyprain}
Let $M$ be a positive integer and $\Sigma$ be a $m \times m$ positive definite symmetric matrix.
We suppose that $E=\{1,...,M\}$ (endowed with the counting measure $\mathfrak m_E$
on $E$) and that the observed process has two components $Y_k=(Z_k,R_k)$. 
For every time $k$, 
$Z_k \in \mathcal Z \subseteq \mathbb R^m$ is a vector of $m$ large scale atmospheric variables (covariates)  and  $R_k \in ([0,+\infty[)^\ell$ is the daily accumulation of rainfall measured at $\ell$ meteorological stations (output time series) with the value $0$ corresponding to dry days.
The model aims at describing the conditional distribution of $\{R_k\}$ given $\{Z_k\}$. For this, we assume that
\begin{equation}
\label{eq:p1rain}
p_{1,\theta}(x_k|x_{k-1},y_{k-1}) =\frac{ q_{x_{k-1},x_k} \exp\left(-1/2 \left(z_{k-1}-\mu_{x_{k-1},x_k}\right)' \Sigma^{-1} 
\left(z_{k-1}-\mu_{x_{k-1},x_k}\right)\right)}
{\sum_{x"=1}^M q_{x_{k-1},x"} \exp\left(-1/2 \left(z_{k-1}-\mu_{x_{k-1},x"}\right)' \Sigma^{-1} 
\left(z_{k-1}-\mu_{x_{k-1},x"}\right)\right)},
\end{equation}
with $q_{x,x'}>0$, $\mu_{x_,x'} \in \mathbb R^m$ 
and (\ref{eq:covariate}) holds
with respect to $\mathfrak m_{\mathcal Z}\otimes \mathfrak m_0^{\otimes \ell}$, where $\mathfrak m_{\mathcal Z}$ is
the Lebesgue measure on $\mathcal Z$ and where $\mathfrak m_{0}$ is the sum of the Dirac
measure $\delta_0$ and of the Lebesgue measure on $(0,+\infty[$.
We observe that $\{Z_k\}_k$ is a Markov chain which transition kernel 
depends neither on the current weather type 
nor on the unknown parameter $\theta$ (typically $Z_k$ is 
the output of an atmospheric model and is considered as an input to the Markov switching model) 
and that the conditional distribution of $R_k$ given $X_k$ and $\{Y_{k'}\}_{k'<k}$ only depends 
on $X_k$ as in usual HMMs. 
Finally the rainfall at the different locations is assumed to be 
conditionally independent given the weather type 
\begin{equation}
\nonumber
p_{R,\theta}\left(r_k(1),...,r_k(l)|x_k\right) = \prod_{i=1}^\ell p_{R_i,\theta}\left(r_k(i)|x_k\right)
\end{equation}
and the rainfall at the different locations is given by the product of Bernoulli and Gamma variables
\begin{equation}
p_{R_i,\theta}(r_{k}(i)|x_{k}) = \left\{\begin{array}{ll}
1-\pi_{i}^{(x_{k})} & (r_{k}(i)=0) \\
\pi_{i}^{(x_{k})}\gamma(r_{k}(i);\alpha_{i}^{(x_{k})},\beta_{i}^{(x_{k})})
& (r_{k}(i)>0)
\end{array}\right .
\end{equation}
where $0 < \pi_{i}^{(x)}< 1$, $\alpha_{i}^{(x)}>0$, $\beta_{i}^{(x)}>0$ and 
 $\gamma(.;\alpha,\beta)$ denotes the pdf of a Gamma distribution with parameters
$\alpha$, $\beta$:
$$\gamma(r;\alpha,\beta)=r^{\alpha-1}\frac{\beta^\alpha e^{-\beta r}}{\Gamma(\alpha)}.$$
The parameter $\theta$ corresponds to
$$\theta=\left((q_{x,x'}),(\mu_{x,x'}),(\pi_i^{(x)}),(\alpha_i^{(x)}),(\beta_i^{(x)})\right).$$
We write $\tilde\Theta$ for the set of such parameters $\theta$ satisfying, for every
$x\in\{1,...,M\}$ and every $i\in\{1,...,\ell\}$,
$$\sum_{x'=1}^M q_{x,x'}=1, \ \ 0<q_{x,x'}<1, \ \sum_{x'=1}^M \mu_{x,x'}=0,\ 0<\pi^{(x)}<1,\ \alpha_{i}^{(x)}>0,
   \ \mbox{and}\ \beta_{i}^{(x)}>0.$$
\end{hypo}
The conditions $\sum_{x'=1}^M q_{x,x'}=1$ and $\sum_{x'=1}^M \mu_{x,x'}=0$ come from \cite{Hughes99}. These conditions are not restrictive. Indeed, 
$q_\theta$ is unchanged if we replace $\mu_{x,x'}$ by $\mu_{x,x'}-\sum_{x"}\mu_{x,x"}$ and $q_{x,x'}$ by $\frac{q_{x,x'}\exp(-(\mu_{x,x'})\Sigma^{-1}\mu_x)}
{\sum_{x"}q_{x,x"}\exp(-(\mu_{x,x"})\Sigma^{-1}\mu_x)}$ (with
$\mu_x:=\sum_{x"}\mu_{x,x"}$).

Observe that the fact that, if $\mu_{x,x'}=0$ for every $x,x'$, then 
$\{X_k\}_k$ is an homogeneous Markov chain and $\{Z_k\}_k$ does not plays any role in the
dynamics of $\{X_k,R_k\}_k$.
\subsection{Properties of the Markov chain}

We start by recalling a classical result ensuring (\ref{HH4}) in the context of HMM (a proof of this result is given 
in Appendix \ref{proofrain} for completeness).
\begin{lemma}[HMM]\label{propHH4}
Fix $\theta$.
Assume that $p_{1,\theta}(x|x',y')=p_{1,\theta}(x|x')$ does not depend on $y'$, $\{X_k\}_k$ is a Markov chain
with transition kernel $Q_{1,\theta}$ admitting an invariant pdf $h_{1,\theta}$ (wrt $\mathfrak m_E$) such that
\begin{equation}
\nonumber
\lim_{n\rightarrow +\infty}\sup_{\nu\in \mathcal P(E)}||[Q_{1,\theta}^{*n}\nu]-h_{1,\theta}||_{L^1(\mathfrak m_E)}=0.
\end{equation}
Assume moreover that $s=0$ (this means that we can take $s=1$ with $p_{2,\theta}(y|x,y')=p_{2,\theta}(y|x)$).
Then  there exists an invariant measure $\nu_\theta$ with pdf $h_\theta$ (wrt $\mathfrak m_E\times \mathfrak m_K$)
given by $h_{\theta}(x,y):=h_{1,\theta}(x)p_{2,\theta}(y|x)$ and 
\begin{equation}\
\nonumber
\lim_{n\rightarrow +\infty}\sup_{\nu\in\mathcal P(E\times K)}||[Q_\theta^{*n}\nu]-h_\theta||_{L^1(\mathfrak m_E\times \mathfrak m_K)}=0.
\end{equation}
Moreover, if $p_{2,\theta}>0$ and if $\{X_k\}_k$ is an aperiodic positive Harris recurrent Markov chain, then
the Markov chain $\{X_k,Y_k\}_k$ is positive Harris recurrent and aperiodic.
\end{lemma}
Due to this lemma, assumption (\ref{HH4}) 
holds true and $\{X_k,Y_k\}_k$ is aperiodic
positive Harris recurrent as soon as $\{X_k,Z_k\}_k$ is aperiodic positive Harris recurrent.

The ergodicity of $\{X_k,Y_k\}_k$ will also follow from the ergodicity of 
$\{X_k,Z_k\}_k$.

\subsection{Consistency of MLE}
\begin{corollary}\label{cororain}
Assume Hypothesis \ref{hyprain}.
Assume that $\Theta$ is a compact subset of $\tilde\Theta$ and that,
for every $\theta\in\Theta$, the transition kernel $Q_{0,\theta}$ 
of the Markov chain $\{X_k,Z_k\}_k$ admits an invariant pdf $h_{0,\theta}>0$ 
(wrt $\mathfrak m_E\times \mathfrak m_{\mathcal Z}$) such that
\begin{equation}
\label{cvgceL1X}
\lim_{n\rightarrow +\infty}\sup_{\nu\in \mathcal P(E\times \mathcal Z)}
\Vert[Q_{0,\theta}^{*n}\nu]-h_{0,\theta}
\Vert_{L^1(\mathfrak m_E\times\mathfrak  m_{\mathcal Z})}=0.
\end{equation}
Assume moreover that $\mathcal Z$ is compact, that
\begin{equation}
\label{pZ1}
\forall z\in\mathcal Z,\ \ \sup_{z_{-1}\in\mathcal Z}p_Z(z|z_{-1})<\infty
\end{equation}
and that
\begin{equation}\label{pZ2}
\bar{\mathbb E}_{\theta^*}[|\log p_Z(Z_0|Z_{-1})|]<\infty.
\end{equation}
Then, for every $x_0\in\{1,...,M\}$, on a set of probability one 
(for $\bar{\mathbb P}_{\theta^*}$),
the limit values $\theta$ of the sequence of random variables 
$(\hat \theta_{n,x_0})_n$ are $\bar{\mathbb P}_{\theta^*}$-almost
surely contained in $\{\theta\in\Theta\ :\ \bar{\mathbb P}_{\theta}=
    \bar{\mathbb P}_{\theta^*}\}$. 

If, moreover, $\{X_k,Z_k\}_k$ is aperiodic and positive Harris recurrent then this result holds true for any initial probability distribution.
\end{corollary}
\begin{proof}
Due to the previous section, we know that (\ref{cvgceL1X}) implies
(\ref{HH4}) and that the aperiodicity and positive Harris recurrence
of $\{X_k,Z_k\}_k$ implies the positive Harris recurrence
of $\{X_k,Y_k\}_k$.

The fact that $\Theta$ is a compact subset of $\tilde\Theta$ directly implies
(\ref{HH1}).

Assumption (\ref{HH3c}) holds true since $E$ is finite, since $p_{R,\theta}(r|x)<\infty$ for 
every $(x,y)\in E\times K$ and according to (\ref{pZ1}).

Now according to (\ref{pZ2}), (\ref{HH3}) and (\ref{HH3b}) will follow from the fact that,
for every $x_0\in X$ and every $i\in\{1,...,\ell\}$,
$$
\bar{\mathbb E}_{\theta^*} \left[\left|\log\left(\inf_\theta p_{R_i,\theta}
        (R_i|x_0)\right)\right|\right]+
\bar{\mathbb E}_{\theta^*} \left[\left|\log\left(\sup_\theta p_{R_i,\theta}
        (R_i|x_0)\right)\right|\right]<\infty.$$
Now we observe that if $R_i=0$, then
$$ 0< 1-\pi_+  \le p_{R_i,\theta}
        (R_i|x_0)\le 1-\pi_-,$$
where $\pi_-$ and $\pi_+$ are the minimal and maximal possible values of $\pi^{(x)}_i$ (for $x\in X$, $i\in\{1,...,\ell\}$
and $\theta$ in the compact set $\Theta$).
Analogously, let us write $\alpha_-$, $\alpha_+$ for the minimal and maximal possible values of $\alpha_{i}^{(x)}$
and $\beta_-$, $\beta_+$ for the minimal and maximal possible values of $\beta_{i}^{(x)}$.
Since, all this quantities are positive and finite, due to the expression of 
$\log( p_{R_i,\theta}(R_i|x_0))$, to prove (\ref{HH3}) and (\ref{HH3b}), it is enough to prove that

$$\bar{\mathbb E}_{\theta^*} [R_i]<\infty\ \ \mbox{and}\ \ \bar{\mathbb E}_{\theta^*} [|\log(R_i)|{\mathbf 1}_{\{R_i>0\}}]<\infty.$$
Observe that, under the stationary distribution, the pdf $h_i$ of $R_i$ satisfies:
$$\forall r>0,\ \  h_i(r)\le (r^{\alpha_--1}{\mathbf 1}_{\{r\le 1\}}+r^{\alpha_+-1}{\mathbf 1}_{\{r> 1\}})
    \frac{\max(\beta_+^{\alpha_+},\beta_+^{\alpha_-}) e^{-r\beta_- }}{\Gamma(\alpha_-)}.$$
Therefore, (\ref{HH3}) and (\ref{HH3b}) come from the facts that $r\mapsto r^{\alpha_+-1}e^{-r\beta_-}$ is integrable at 
$+\infty$ (since $\beta_->0$) and that $r\mapsto |\log r|r^{\alpha_--1}$ is integrable at 0 (since $\alpha_->0$).
\end{proof}
Now we will add an assumption on $\theta$ to ensure the identifiability
of the parameter. 
If we assume $\pi_i^{(x)}=0$ for every $i$ and every $x$, then identifiability
follows easily if we assume moreover that
\begin{equation}
\label{eq:identpasbon}
x\ne x'\ \Rightarrow\ (\alpha_i^{(x)},\beta_i^{(x)})_i\ne  (\alpha_i^{(x')},\beta_i^{(x')})_i.
  \end{equation}
But, if we do not assume $\pi_i^{(x)}=0$, (\ref{eq:identpasbon}) does not
ensure identifiability anymore. We give now an explicit counter-example.
\begin{remark}
Assume $M=\ell=2$. 
We consider two models $A_1$ and $A_2$ associated to $\theta_1$
and $\theta_2$ respectively, with
$$\theta_j=\left((q_{x,x',(j)}),(\mu_{x,x',(j)}),(\pi_i^{(x,(j))}),(\alpha_i^{(x,(j))}),(\beta_i^{(x,(j))})\right),$$

and 
\begin{itemize}
	\item $q_{x,x',(1)}=0.5$, $\mu_{x,x',(1)}=0$, $\pi_i^{(x,(1))}=0.5$, 
$\alpha_i^{(x,(1))}=1$, $\beta_1^{(x,(1))}=1$, $\beta_2^{(1,(1))}=2$, 
$\beta_2^{(2,(1))}=3$, 
\item $q_{x,1,(2)}=0.6$, $q_{x,2,(2)}=0.4$, $\mu_{x,x',(2)}=0$, 
$\pi_1^{(x,(2))}=0.5$, $\pi_2^{(1,(2))}=\frac{0.25}{0.6}$, 
 $\pi_2^{(2,(2))}=\frac{0.25}{0.4}$, 
$\alpha_i^{(x,(2))}=1$, $\beta_1^{(x,(1))}=1$, $\beta_2^{(1,(1))}=2$, 
$\beta_2^{(2,(1))}=3$.
\end{itemize}

For model $A_1$ (under the stationary measure), $\{X_k\}$ is an iid sequence on $\{1,2\}$ with
$\mathbb P(X_1=1)=0.5$ and the distribution of $R_k$ given $\{X_k=1\}$ is
$(0.5\delta_0+0.5\Gamma(1,1))\otimes(0.5\delta_0+0.5\Gamma(1,2)) $
whereas the distribution of $R_k$ taken $\{X_k=2\}$ is
$(0.5\delta_0+0.5\Gamma(1,1))\otimes(0.5\delta_0+0.5\Gamma(1,3)).$
Hence, for the model $A_1$, the $R_k$ are iid with distribution
\begin{equation}\label{loiY}
(0.5\delta_0+0.5\Gamma(1,1))\otimes(0.5\delta_0+0.25\Gamma(1,2)+0.25
\Gamma(1,3)).
\end{equation}
For model $A_2$ (under the stationary measure), $\{X_k\}$is an iid sequence on $\{1,2\}$ with
$\mathbb P(X_1=1)=0.6$ and the distribution of $R_k$ given $\{X_k=1\}$ is
$(0.5\delta_0+0.5\Gamma(1,1))\otimes\left((1-\frac{0.25}{0.6})\delta_0+
\frac{0.25}{0.6}\Gamma(1,2)\right) $
whereas the distribution of $R_k$ taken $\{X_k=2\}$ is
$(0.5\delta_0+0.5\Gamma(1,1))\otimes\left((1-\frac{0.25}{0.4})\delta_0+
\frac{0.25}{0.4}\Gamma(1,3)\right)$.
Hence, for the model $A_2$, the $R_k$ are iid with distribution
(\ref{loiY}). 

Observe that the distribution of $\{Y_k\}$ under the stationary measure
is the same for models $A_1$ and $A_2$.
\end{remark}
The next result (proved in appendix \ref{proofidentrain}) states that the following condition ensures identifiability
\begin{equation}\label{condiident}
x\ne x'\ \Rightarrow\ \forall i\in\{1,...,\ell\},\ (\alpha_{i,\theta_1}^{(x)},\beta_{i,\theta_1}^{(x)})=(\alpha_{i,\theta_1}^{(x')},\beta_{i,\theta_1}^{(x')}).
\end{equation}
\begin{prop}\label{identrain}
Assume Hypothesis \ref{hyprain}. Let $\theta_1$ and $\theta_2$ in 
$\tilde\Theta$, with
$$\theta_j=\left((q_{x,x',(j)}),(\mu_{x,x',(j)}),(\pi_i^{(x,(j))}),(\alpha_i^{(x,(j))}),(\beta_i^{(x,(j))})\right).$$
Assume that $\theta_1$ satisfies (\ref{condiident}).

Then
$\bar {\mathbb P}_{\theta_1}^Y=\bar {\mathbb P}_{\theta_2}^Y$ if and only $\theta_1$ and $\theta_2$ are equal up to a permutation of indices, i.e. there exists a permutation $\tau$ of $\{1,...,M\}$ such that, for every $x,x'\in\{1,...,M\}$ and every $i\in\{1,...,\ell\}$, we have 
$q_{x,x',(1)}=q_{\tau(x),\tau(x'),(2)}$, $\mu_{x,x',(1)}=\mu_{\tau(x),\tau(x'),(2)}$, $\pi_i^{(x,(1))}=\pi_i^{(\tau(x),(2))}$,
$\alpha_i^{(x,(1))}=\alpha_i^{(\tau(x),(j))}$, $\beta_i^{(x,(1))}=\beta_i^{(\tau(x),(2))}$.
\end{prop}
Now the following result is a direct consequence of Corollary \ref{cororain}
and Proposition \ref{identrain}.
\begin{theorem}\label{thmrainfall}
Assume Hypothesis \ref{hyprain}. 
Assume that $\Theta$ is a compact subset of $\tilde\Theta$ and that,
for every $\theta\in\Theta$, the transition kernel $Q_{0,\theta}$ 
of the Markov chain $(X_k,Z_k)_k$ admits an invariant pdf $h_{0,\theta}$ 
(wrt $\mathfrak m_E\times\mathfrak  m_{\mathcal Z}$) satisfying (\ref{cvgceL1X}).
Assume that $\theta^*$ satisfies (\ref{condiident}).
Assume moreover that $\mathcal Z$ is compact, that (\ref{pZ1}) and
(\ref{pZ2}) hold true.
Then, for every $x_0\in\{1,...,M\}$, on a set of probability one 
(for $\bar{\mathbb P}_{\theta^*}$),
the limit values $\theta$ of the sequence of random variables 
$(\hat \theta_{n,x_0})_n$ are equal to $\theta^*$ 
up to a permutation of indices.

If, moreover, $(X_k,Z_k)_k$ is aperiodic and positive Harris recurrent then this result holds true 
for any initial probability distribution.
\end{theorem}

\section{Conclusions}
In this work, we have extended the consistency result of \cite{DMR} to the non-homogeneous case and we have relaxed some other 
of their assumptions
(namely on $p_2$). We have illustrated our results by two specific but representative models for which we gave general conditions
ensuring the consistency of the maximum likelihood estimator. 
Our results opens perspectives in different directions: theoretical results (such as the asymptotic normality of the MLE), 
applied statistics (namely the study of other non-homogeneous switching Markov models and their applications), 
but also the development of a R package
to make easier the practical use of these flexible models.

\begin{appendix}
\section{Consistency : proof of Theorem \ref{consistancegene}}
\label{append2}
As usual, we define the associated transition operator $Q_\theta$ as an operator acting on the set of 
bounded measurable
functions of $E\times K^s$ (it may also act on other Banach spaces $\mathcal B$) by
\begin{eqnarray*}
Q_\theta g(x_0,y_{-s+1}^{0})&=&{\mathbb E}_\theta[g(X_1,Y_{-s+2}^1)|X_0=x_0,Y_{-s+1}^0=y_{-s+1}^0]\\
   &=&\int_{E\times K} g(x_1,y_{-s+2}^1)q_\theta(x_1,y_1|x_{0},y_{-s+1}^{0})\, d\mu_0(x_1,y_1).
\end{eqnarray*}

We denote by $Q_\theta^*$ the adjoint operator of $Q_\theta$ 
defined on $\mathcal B '$ the dual space of $\mathcal B$ 
(if $Q_\theta$ acts on $\mathcal B$) 
by
$$\forall \nu\in\mathcal B ',\ \forall f\in\mathcal B,\ \ Q_\theta^*(\nu)(f)=\nu(Q_{\theta}(f)) .$$
For every integer $k\ge 0$, the measure $(Q_\theta^*)^k(\nu)$ corresponds to the distribution
of $(X_k,Y_{k-s+1}^k)$ if $\{X_l,Y_l\}_l$ is the Markov chain with transition 
operator $Q_\theta$ such that the distribution of $(X_0,Y_{-s+1}^0)$ is $\nu$. 

If $\nu\in\mathcal B'$ has a pdf $h$ with respect to $\mu:=\mathfrak m_E\times\mathfrak  m_K^{\otimes s}$, then $Q_\theta^*\nu$ is also 
absolutely continuous with respect to $\mu$ and its pdf, written $Q_\theta^*h$, is given by
\begin{equation}
\nonumber
 Q_\theta^*h(x_0,y_{-s+1}^0) :=   \int_{E\times K}
       q_\theta(x_0,y_0|x_{-1},y_{-s}^{-1}) h(x_{-1},y_{-s}^{-1})\, d\mu_0(x_{-1},y_{-s}).
\end{equation}
Observe that, due to the particular form of $q_\theta$, for every integer $k\ge s$ and every $P=(x_{-k},y_{-k-s+1}^{-k})
\in E\times K^s$, the measure
$(Q^*_\theta)^k\delta_P$ (where $\delta_P$ is the Dirac measure at $P$) is absolutely continuous with respect
to $\mu:=\mathfrak m_E\times\mathfrak  m_K^{\otimes s}$; its pdf $Q_\theta^{*k}(\cdot|P)$ is given by
$$ Q_\theta^{*k}(x_0,y_{-s+1}^0|P)
   =\int_{E^{k-1}\times K^{k-s}}\prod_{i=1-k}^{0}q_\theta(x_i,y_i|x_{i-1},y_{i-s}^{i-1})\, 
     d\mathfrak m_E^{\otimes (k-1)}
       (x_{-k+1}^{-1})d\mathfrak m_K^{\otimes (k-s)}(y_{-k+1}^{-s}).$$
More generally, for every initial measure $\nu$ and every $k\ge s$, $Q_\theta^{*k}\nu$ is absolutely continuous
with respect to $\mu$ and its pdf $[Q_\theta^{*k}\nu]$ is given by
\begin{equation}\label{Q*mes}
 [Q_\theta^{*k}\nu](\cdot) = \int_{E\times K^s} Q_\theta^{*k}(\cdot|P)\, d\nu(P)  .
\end{equation}
We suppose that, for every $\theta\in\Theta$, there exists an invariant probability measure $\bar\nu_\theta$
for $Q_\theta^*$. Observe that, due to \eqref{Q*mes}, $\bar\nu_\theta$
admits a pdf $h_\theta$ with respect to $\mu$.

We identify $(X_k,Y_k)_{k}$ with the canonical Markov chain $\{(X_0,Y_0)\circ \tau^k\}_k$ defined 
on $\Omega_+:=(E\times K)^{\mathbb N}$ by
$X_0((x_k,y_k)_{k})=x_0$, $Y_0((x_k,y_k)_{k})=y_0$,
$\tau_+$ being the shift 
($\tau_+((x_k,y_k)_k)=(x_{k+1},y_{k+1})_k$). 
We endow $\Omega_+$ with its Borel $\sigma$-algebra $\mathcal F_+$. We 
denote by $\bar{\mathbb P}_\theta$
the probability measure on $(\Omega_+,\mathcal F_+)$ associated to the invariant measure
$\bar\nu_\theta$ and by $\bar{\mathbb E}_\theta$ the corresponding expectation.
The ergodicity of  $(X_k,Y_k)_{k}$ is equivalent to the ergodicity of $(\Omega,\mathcal F,\bar{\mathbb P}_{\theta^*},\tau)$.

We now follow and adapt the proof
of \cite[Thm. 1]{DMR} (see Lemmas \ref{lem5} and \ref{lem7}).
We do not give all the details of the proofs when they are a direct rewriting of \cite{DMR}.
First, we consider the stationary case.
Let $\tau$ be the full shift on $\Omega:=(E\times K)^{\mathbb Z}$.
For every $k\in\mathbb Z$, we identify $X_k$ with $X_0\circ \tau^k$ and $Y_k$ with
$Y_0\circ \tau^k$, where $X_0((x_m,y_m)_{m\in\mathbb Z}):=x_0$ and $Y_0((x_m,y_m)_{m\in\mathbb Z})=y_0$.
\subsection{Likelihood and stationary likelihood}
We start by recalling a classical fact in the context of Markov chains (and the proof
of which is direct).
\begin{fact}\label{CMinh}
Let $m$ and $n$ belong to $\mathbb Z$ with $m\le n$.
Under $\bar {\mathbb P}_\theta$, conditionally to $(Y_{m-s+1}^n)$, $(X_k)_{k\in\{m,...,n\}}$ is a 
(possibly nonhomogeneous) Markov chain. 
Moreover, under $\bar {\mathbb P}_\theta$, the conditional pdf (wrt $\mathfrak m_E$) of $X_k$ given 
$(X_m^{k-1},Y_{m-s+1}^n)$
is given by
\begin{equation}\label{facteq1}
 {p}_\theta(X_k=x_k|X_m^{k-1},Y_{m-s+1}^n)
     =\frac{p_\theta(Y_{k}^n,X_k=x_k|X_{k-1},Y_{k-s}^{k-1})}
         {p_\theta(Y_{k}^n|X_{k-1},Y_{k-s}^{k-1})}\ \ \bar{\mathbb P}_\theta-a.s.,
\end{equation}
with
\begin{equation}\label{facteq2}
 p_\theta(Y_{k}^n,X_k=x_k|X_{k-1}=x_{k-1},Y_{k-s}^{k-1})
      :=\int_{E^{n-k}}\prod_{j=k}^nq_\theta(x_j,Y_j|x_{j-1},Y_{j-s}^{j-1})\, d\mathfrak m_E^{\otimes (n-k)}(x_{k+1}^n)
\end{equation}
and
\begin{equation}\label{facteq3}
  p_\theta(Y_{k}^n|X_{k-1},Y_{k-s}^{k-1})
      :=\int_E  p_\theta(Y_{k}^n,X_k=x_k|X_{k-1},Y_{k-s}^{k-1})\, d\mathfrak m_E(x_k).
\end{equation}
\end{fact}
Using (\ref{HH1}), (\ref{HH3}) and (\ref{HH3b}), we observe that
the quantities appearing in this fact are well-defined.
Due to Fact \ref{CMinh}, the quantity
$\bar {p}_\theta(X_k=x_k|X_{k-1},Y_{m-s+1}^n)$ is equal to
$$\frac{\int_{E^{n-k+1}} (\prod_{j=k+1}^na_j)
    p_{1,\theta}(\tilde x_k|X_{k-1},Y_{k-s}^{k-1})
        p_{2,\theta}(Y_k|\tilde x_k,Y_{k-s}^{k-1}) \, d\delta_{x_k}(\tilde x_k)\, d\mathfrak m_E^{\otimes (n-k)}(\tilde x_{k+1}^n)}
    {\int_{E^{n-k+1}}(\prod_{j=k+1}^na_j)
         p_{1,\theta}(\tilde x_k|X_{k-1},Y_{k-s}^{k-1})  
     p_{2,\theta}(Y_k|\tilde x_k,Y_{k-s}^{k-1}) \, d\mathfrak m_E^{\otimes (n-k+1)}(\tilde x_{k}^n)},$$
with $a_j:=q_\theta(\tilde x_j,Y_j|\tilde x_{j-1},Y_{j-s}^{j-1})$.
Therefore
\begin{equation}\label{toto}
\bar {p}_{\theta}(X_k=x_k|X_{k-1},Y_{m-s+1}^n)\ge
\frac{ p_{1,-}}{p_{1,+}}
    \beta(x_k),\ \mbox{with}\ \ 
    \beta(x_k):=\frac{p_{\theta}(Y_k^n|X_k=x_k,Y_{k-s}^{k-1})}
      {\int_{E}p_{\theta}(Y_k^n|X_k=\tilde x_k,Y_{k-s}^{k-1})\, d\mathfrak m_E(\tilde x_k)}.
\end{equation}
{}From this last inequality (since $0<p_{1,-}<p_{1,+}<\infty$), we directly get the following (from \cite{Lindvall}).
\begin{corollary}(as \cite[Cor. 1]{DMR})\label{coro}
For all $m\le k\le n$ and every probability measures $\mathfrak m_1$ and $\mathfrak m_2$ on $E$, we have, 
$\bar{\mathbb P}_\theta-a.s.$
$$\left\Vert \int_E\bar {\mathbb P}_\theta(X_k\in\cdot|X_m=x_m,Y_{m-s+1}^n)\, d\mathfrak m_1(x_m)
            - \int_E\bar {\mathbb P}_\theta(X_k\in\cdot|X_m= x_m,Y_{m-s+1}^n)\, d\mathfrak m_2(x_m)  \right\Vert_{TV}\le
     \rho^{k-m},$$
with $\rho:=1-\frac{p_{1,-}}{p_{1,+}}$.
\end{corollary}
Observe that the log-likelihood $\ell_n(\theta,x_0)$ satisfies
\begin{equation}
\nonumber
\ell_n(\theta,x_0)=\sum_{k=1}^n\log  p_\theta(Y_k|X_0=x_0,Y_{-s+1}^{k-1})
\ \ \bar{\mathbb P}_\theta-a.s.,
\end{equation}
with

$\displaystyle p_\theta(Y_k|X_0=x_0,Y_{-s+1}^{k-1})
     :=\frac{p_\theta(Y_1^k|X_0=x_0,Y_{-s+1}^0)}{p_\theta(Y_1^{k-1}|X_0=x_0,Y_{-s+1}^0)}$
$$=\int_{E^2}q_\theta(x_k,Y_k|x_{k-1},Y_{k-s}^{k-1})
    p_\theta(X_{k-1}=x_{k-1}|X_0=x_0,Y_{-s+1}^{k-1})\, d\mathfrak m_E^{\otimes 2}(x_k,x_{k-1}).$$
Let us now define
the stationary log-likelihood $\ell_n(\theta)$ by
\begin{equation}
\nonumber
\ell_n(\theta):=\sum_{k=1}^n\log\bar {p}_{\theta}(Y_k|Y_{-s+1}^{k-1}),
\end{equation}
with
$$\bar {p}_{\theta}(Y_k|Y_{-s+1}^{k-1}):=
    \int_{E^2}q_\theta(x_k,Y_k|x_{k-1},Y_{k-s}^{k-1})
    \bar{p}_\theta(X_{k-1}=x_{k-1}|Y_{-s+1}^{k-1})\, d\mathfrak m_E^{\otimes 2}(x_k,x_{k-1})$$
and
$$\bar{p}_\theta(X_{k-1}=x_{k-1}|Y_{-s+1}^{k-1})
      :=\int_E p_\theta(X_{k-1}=x_{k-1}|X_0=x_0,Y_{-s+1}^{k-1})
         \bar p_\theta(X_0=x_0|Y_{-s+1}^{k-1})\, d\mathfrak m_E(x_0).$$
\begin{lemma}(as \cite[Lem. 2]{DMR}) We have
\begin{equation}\label{(**)}
\sup_{x_0\in E}\sup_{\theta\in\Theta}|\ell_n(\theta,x_0)-\ell_n(\theta)|\le \frac 1{(1-\rho)^2}\ \ 
       \ \bar{\mathbb P}_{\theta^*}-a.s.,
\end{equation}
\end{lemma}
\begin{proof}
We have

$\displaystyle\sup_{x_0\in E}| {p}_{\theta}(Y_k|X_0=x_0,Y_{-s+1}^{k-1})
   - \bar {p}_{\theta}(Y_k|Y_{-s+1}^{k-1})|\le$
$$\le p_{1,+}\int_{E^3}p_{2,\theta}(Y_k|x_k,Y_{k-s}^{k-1})D(x_{k-1},x_0,x)
   \bar p_\theta(X_0=x|Y_{-s+1}^{k-1})\, d\mathfrak m_E^{\otimes 3}(x,x_{k-1},x_k), $$
with
$D(x_{k-1},x_0,x):= |  p_\theta(X_{k-1}=x_{k-1}|X_0=x_0,Y_{-s+1}^{k-1})-
     p_\theta(X_{k-1}=x_{k-1}|X_0=x,Y_{-s+1}^{k-1})|.$
Due to Corollary \ref{coro}, we have
$$|{p}_{\theta}(Y_k|X_0=x_0,Y_{-s+1}^{k-1})
   - \bar {p}_{\theta}(Y_k|Y_{-s+1}^{k-1})|\le p_{1,+}\rho^{k-1}\int_{E}p_{2,\theta}(Y_k|x_k,Y_{k-s}^{k-1})\, 
   d\mathfrak m_E(x_k). $$
Since $|{p}_{\theta}(Y_k|X_0,Y_{-s+1}^{k-1})|$ and $|{p}_{\theta}(Y_k|Y_{-s+1}^{k-1})|$ are both larger than
or equal to
$$ p_{1,-} \int_{E}p_{2,\theta}(Y_k|x_k,Y_{k-s}^{k-1})\, 
   d\mathfrak m_E(x_k),$$
we obtain that
\begin{eqnarray}
\left|\log {p}_{\theta}(Y_k|X_0=x_0,Y_{-s+1}^{k-1})
   - \log\bar {p}_{\theta}(Y_k|Y_{-s+1}^{k-1})\right|
&\le& \frac{|{p}_{\theta}(Y_k|X_0=x_0,Y_{-s+1}^{k-1})
    -\bar {p}_{\theta}(Y_k|Y_{-s+1}^{k-1})|}{p_{1,-}\int_E
     p_{2,\theta}(Y_k|x_k,Y_{k-s}^{k-1})\, 
   d\mathfrak m_E(x_k)}\nonumber \\
&\le&\rho^{k-1}\frac{p_{1,+}}{p_{1,-}}=
   \frac{\rho^{k-1}}{1-\rho}\ \ \bar{\mathbb P}_\theta-a.s.\label{majo}
\end{eqnarray}
and so (\ref{(**)})
since $\bar{\mathbb P}_{\theta^*}$ is absolutely continuous
with respect to $\bar{\mathbb P}_{\theta}$ (for all $\theta$).
\end{proof}
\subsection{Asymptotic behavior of the log-likelihood}
The idea is to approximate $n^{-1}\ell_n(\theta)$ by 
$n^{-1}\sum_{k=1}^n\log
  {p}_\theta(Y_k|Y_{-\infty}^{k-1}).$
To this end, we define, for any $k\ge 0$, any $m\ge 0$ and any $x_0\in E$,
the following quantities
$$\Delta_{k,m,x}(\theta):=\log\bar{p}_\theta(Y_k|Y_{-m-s+1}^{k-1},X_{-m}=x)
\ \ \mbox{and}\ \ 
\Delta_{k,m}(\theta):=\log\bar{p}_\theta(Y_k|Y_{-m-s+1}^{k-1}).$$
With these notations, we have 
\begin{equation}\label{logvrais}
\ell_n(\theta)=\sum_{k=1}^n\Delta_{k,0}(\theta)\ \mbox{ and }\ 
\ell_n(\theta,x_0)=\sum_{k=1}^n\Delta_{k,0,x_0}(\theta).
\end{equation}
\begin{lemma}(as \cite[Lemma~3]{DMR})\label{lemme3}
With the notation $\rho$
introduced in Corollary \ref{coro}, we have $\bar{\mathbb P}_{\theta^*}$-almost surely
\begin{equation}\label{eq1}
\forall m,m'\ge 0,\ \ \sup_{\theta\in\Theta}\sup_{x,x'\in E}|\Delta_{k,m,x}(\theta)-\Delta_{k,m',x'}(\theta)|\le
       \rho^{k+\min(m,m')-1}/(1-\rho)
\end{equation}
\begin{equation}\label{eq2}
\forall m\ge 0,\ \ \sup_{\theta\in\Theta}\sup_{x\in E}|\Delta_{k,m,x}(\theta)-\Delta_{k,m}(\theta)|
     \le  \rho^{k+m-1}/(1-\rho)
\end{equation}
\begin{equation}\label{eq3}
\sup_\theta\sup_{m\ge 0}\sup_{x\in E}|\Delta_{k,m,x}(\theta)|\le\max(|\log(p_{1,+}b_+(Y_{k-s}^k))|,
  |\log(p_{1,-} b_-(Y_{k-s}^k))|)
\end{equation}
with
$$ b_-(y_{k-s}^k):=\inf_\theta \int_E p_{2,\theta}(y_k|x,y_{k-s}^{k-1})\, d\mathfrak m_E(x)$$
and\ \ 
$$ b_+(y_{k-s}^k):=\sup_\theta \int_E p_{2,\theta}(y_k|x,y_{k-s}^{k-1})\, d\mathfrak m_E(x).$$
\end{lemma}
\begin{proof}
Assume that $m\le m'$.
We have
$$e^{\Delta_{k,m,x}(\theta)}=\int_{E^2}q_\theta(x_k,Y_k|x_{k-1},Y_{k-s}^{k-1})p_\theta
      (X_{k-1}=x_{k-1}|X_{-m}=x,Y_{-m-s+1}^{k-1})\, d\mathfrak m_E^{\otimes 2}(x_k,x_{k-1}).$$
Observe moreover that, due to Fact \ref{CMinh}, we have
$$ e^{\Delta_{k,m',x'}(\theta)}=
   \int_E e^{\Delta_{k,m,x''}(\theta)} p_\theta
      (X_{-m}=x''|X_{-m'}=x',Y_{-m'-s+1}^{k-1})\, d\mathfrak m_E(x'').$$
Therefore, according to Corollary \ref{coro}, we obtain
\begin{eqnarray*}
\left|e^{\Delta_{k,m,x}(\theta)}-e^{\Delta_{k,m',x'}(\theta)}\right|
&\le& \sup_{x"\in E}|e^{\Delta_{k,m,x}(\theta)}-e^{\Delta_{k,m,x"}(\theta)}|\\
&\le& p_{1,+}\rho^{k+m-1}\int_E p_{2,\theta}(Y_k|x_k,Y_{k-s}^{k-1})\, d\mathfrak m_E(x_k).
\end{eqnarray*}
Since 
$$\left|e^{\Delta_{k,m,x}(\theta)}\right|\ge p_{1,-}
\int_E p_{2,\theta}(Y_k|x_k,Y_{k-s}^{k-1})\, d\mathfrak m_E(x_k),$$
we get the first point.
The proof of the second point follows exactly the same scheme with the use of the following formula
$$e^{\Delta_{k,m}(\theta)}= \int_{E}e^{\Delta_{k,m,x_{-m}}(\theta)}
    \bar p_{\theta}(X_{-m}=x_{-m}|Y_{-m-s+1}^{k-1})\, d\mathfrak m_E(x_{-m}).$$
The last point comes from the fact that
$$p_{1,-} \int_E p_{2,\theta}(Y_k|x_k,Y_{k-s}^{k-1})\, d\mathfrak m_E(x_k)\le e^{\Delta_{k,m,x}(\theta)}\le 
  p_{1,+} \int_E p_{2,\theta}(Y_k|x_k,Y_{k-s}^{k-1})\, d\mathfrak m_E(x_k).$$
\end{proof}
Due to (\ref{eq1}), we get that, $\bar{\mathbb P}_{\theta^*}$-a.s.,
$(\Delta_{k,m,x}(\theta))_m$ is a (uniform in $(k,x,\theta)$) Cauchy sequence and so
converges uniformly in $(k,x,\theta)$ to some $\Delta_{k,\infty,x}(\theta)$.

Due to (\ref{eq1}) and (\ref{eq2}), $\Delta_{k,\infty,x}(\theta)$ does not depend on $x$ and
will be denoted by $\Delta_{k,\infty}(\theta)$. Moreover we have $\Delta_{k,\infty}(\theta)
=\Delta_{0,\infty}(\theta)\circ\tau^k$. 

Due to (\ref{eq3}), (\ref{HH1}), (\ref{HH3}) and (\ref{HH3b}),
$(\Delta_{k,m,x}(\theta))_{k,m,x}$ is uniformly bounded in 
$\mathbb L^1(\bar{\mathbb P}_{\theta^*})$. Therefore $\Delta_{k,\infty}(\theta)$
is in $\mathbb L^1(\bar{\mathbb P}_{\theta^*})$. Let us write
$$\ell(\theta):= \bar{\mathbb E}_{\theta^*}[\Delta_{0,\infty}(\theta)].$$
Since $(\Omega,\mathcal F,\bar{\mathbb P}_{\theta^*},\tau)$ is ergodic,
from the Birkhoff-Khinchine ergodic theorem, we have
\begin{equation}\label{birkhoff}
\lim_{n\rightarrow +\infty}n^{-1}\sum_{k=1}^n\Delta_{k,\infty}(\theta)=
       \ell(\theta)   \ \ \ \bar{\mathbb P}_{\theta^*}-a.s.\
    \mbox{ and in }\mathbb L^1(\bar{\mathbb P}_{\theta^*}).
\end{equation}
Now, due to (\ref{eq1}) and (\ref{eq2}) applied with $m=0$, we obtain
\begin{equation}\label{delta0}
\sum_{k=1}^n\sup_\theta|\Delta_{k,0}(\theta)-\Delta_{k,\infty}(\theta)|
    \le \frac {2}{(1-\rho)^2}\ \ \ \bar{\mathbb P}_{\theta^*}-a.s..
\end{equation}
Now, putting together (\ref{logvrais}), (\ref{eq2}), (\ref{birkhoff}) and (\ref{delta0}), we have
\begin{corollary}
$$
\lim_{n\rightarrow+\infty}n^{-1}\ell_n(\theta,x_0)=
   \lim_{n\rightarrow+\infty}n^{-1}\ell_n(\theta)= \ell(\theta),\ \ 
\bar{\mathbb P}_{\theta^*}-a.s.. $$
\end{corollary}
Still following \cite{DMR}, we have the next lemma insuring the continuity of $\theta\mapsto
\ell(\theta)$.
\begin{lemma}(as \cite[Lemma 4]{DMR})\label{lemme4}
For all $\theta\in\Theta$,
$$\lim_{\delta\rightarrow 0}\bar{\mathbb E}_{\theta^*}[\sup_{|\theta-\theta'|\le\delta}
       |\Delta_{0,\infty}(\theta)-\Delta_{0,\infty}(\theta')|]=0. $$ 
\end{lemma}
\begin{proof}
We recall that $\Delta_{0,\infty}=\lim_{m\rightarrow\infty}\Delta_{0,m,x}(\theta)$ (for every $x\in E$) with
$$\Delta_{0,m,x}(\theta)=\log\frac{\int_{E^m}\prod_{\ell=-m+1}^0q_\theta(x_\ell,Y_\ell|x_{\ell-1},
       Y_{\ell-s}^{\ell-1})\, dm_E^{\otimes m}(x_{-m+1}^0)\, d\delta_x(x_{-m})}{\int_{E^{m-1}}\prod_{\ell=-m+1}^{-1}
       q_\theta(x_\ell,Y_\ell|x_{\ell-1},
       Y_{\ell-s}^{\ell-1})\, dm_E^{\otimes (m-1)}(x_{-m+1}^{-1})\, d\delta_x(x_{-m})} .$$
Since the maps $\theta\mapsto q_\theta(x_\ell,y_\ell|x_{\ell-1},
       y_{\ell-s}^{\ell-1},y_{\ell})$ are continuous, $\Delta_{0,m,x}$ is 
$\bar{\mathbb P}_{\theta^*}$-almost surely continuous.
The uniform convergence result proved above insures that $\Delta_{0,\infty}$ is also 
$\bar{\mathbb P}_{\theta^*}$-almost surely continuous.
Hence
$$\forall\theta,\ \ \lim_{\delta\rightarrow 0} \sup_{\theta':|\theta-\theta'|\le\delta}
       |\Delta_{0,\infty}(\theta)-\Delta_{0,\infty}(\theta')|=0 \ \ \bar{\mathbb P}_{\theta^*}-a.s..$$
Now, the result follows from the Lebesgue dominated convergence theorem, due to (\ref{eq3}), (\ref{HH1}),
(\ref{HH3}) and (\ref{HH3b}).
\end{proof}
\begin{lemma}(as \cite[Prop. 2]{DMR})\label{prop2}
We have
$$\lim_{n\rightarrow +\infty}\sup_{\theta\in\Theta}|n^{-1}\ell_n(\theta,x_0)-\ell(\theta)|=0,
    \ \ \ \bar{\mathbb P}_{\theta^*}-a.s..$$
\end{lemma}
Lemma \ref{prop2} can be deduced exactly as in the proof of \cite[Prop. 2]{DMR}.
We do not rewrite the proof, but mention that it uses (\ref{(**)}), the compacity of $\Theta$,
the continuity of $\ell$, (\ref{delta0}), the ergodicity of $\bar{\mathbb P}_{\theta^*}$
and Lemma \ref{lemme4}.
\begin{lemma}\label{lem5}(as \cite[Lemma 5]{DMR})
For every $k\le \ell$, we have
$$\lim_{j\rightarrow -\infty}\sup_{i\le j}
   |\bar{p}_\theta(Y_k^\ell|Y_{i-s+1}^j)-\bar{p}_\theta(Y_k^\ell)|=0\ \ 
   \mbox{in}\ \bar{\mathbb P}_{\theta^*}-probability .$$
\end{lemma}
\begin{proof}
Let us write
$G(y_{-s}^0):=\int_Ep_{2,\theta}(y_0|x,y_{-s}^{-1})\, d\mathfrak m_E(x)$ and 
$\tilde G(y_0):=\sup_{y_{-s}^{-1}}G(y_{-s}^0)$.
As in the proof of \cite[Lemma 5]{DMR}, we observe that, by stationarity, it is enough to prove that
$$\forall \ell>0,\ \ \ \lim_{k\rightarrow +\infty}\sup_{i\ge 0}
     \left|\bar p_\theta(Y_{k}^{k+\ell}|Y_{-i-s+1}^{0})
      - \bar p_\theta(Y_{k}^{k+\ell})\right|=0\ \ in\ \bar {\mathbb P}_{\theta^*}-probability$$
and we write
$$\left|\bar p_\theta(Y_{k}^{k+\ell}|Y_{-i-s+1}^{0})
      - \bar p_\theta(Y_{k}^{k+\ell})\right|=\left |\int_{E^2\times K^{2s}}A_k(B'_k-B''_k)C_i
\,   d\mathfrak m_E^{\otimes 2}(x_s,x_{k-1})d\mathfrak m_K^{\otimes 2s}(y_1^s,y_{k-s}^{k-1})\right|,$$
with
$$A_k:= p_\theta(Y_k^{k+\ell}|X_{k-1}=x_{k-1},Y_{k-s}^{k-1}=y_{k-s}^{k-1}) \le \tilde A_k:= 
    p_{1,+}^{\ell+1}\prod_{j=k+s}^{k+\ell} G(Y_{j-s}^j)\prod_{j=k}^{k+s-1} \tilde G(Y_j),$$
(due to (\ref{facteq3}) and to (\ref{HH1})) with
$$B'_k:= p_\theta(X_{k-1}=x_{k-1},Y_{k-s}^{k-1}=y_{k-s}^{k-1}|X_s=x_s,Y_1^s=y_1^s)
    =Q_\theta^{*(k-s-1)}(x_{k-1},y_{k-s}^{k-1}|x_s,y_1^s) ,$$
with
$$B''_k:=\bar p_\theta(X_{k-1}=x_{k-1},Y_{k-s}^{k-1}=y_{k-s}^{k-1})=h_\theta(x_{k-1},y_{k-s}^{k-1}) $$
and with
$$C_i:=\bar{p}_\theta (X_{s}=x_{s},Y_{1}^{s}=y_{1}^{s}|Y_{-i-s+1}^0).$$
Let us write
$$B_k:=\int_{E\times K^s}|B'_k-B''_k|\, d\mu(x_{k-1},y_{k-s}^{k-1}).$$
We have
$$\left|\bar p_\theta(Y_{k}^{k+\ell}|Y_{i-s+1}^{0})
      - \bar p_\theta(Y_{k}^{k+\ell})\right|\le \tilde A_k\int_{E\times K^s}B_kC_i\, d\mu(x_s,y_{1}^s).$$
On the one hand, due to (\ref{HH4}), $B_k=B_k(x_s,y_1^s)$ converges to 0 as $k$ goes to infinity, for 
$\mu$-almost every $(x_s,y_1^s)$ (and this quantity is bounded by 1).
On the other hand, on $\{Y_{-i-s+1}^0=y_{-i-s+1}^0\}$, we have
\begin{eqnarray*}
C_i&=&\int_{E^s}\prod_{j=1}^s q_\theta(x_j,y_j|x_{j-1},y_{j-s}^{j-1})
\bar{p}_\theta(X_0=x_0|Y_{-i-s+1}^0=
       y_{-i-s+1}^0)\, d\mathfrak m_E^{\otimes s}(x_0^{s-1})\\
&\le& p_{1,+}H(x_s,y_{-s+1}^s),
\end{eqnarray*}
with 
$$H(x_s,y_{-s+1}^s):=\int_{E^{s-1}}\prod_{j=2}^s p_{1,\theta}(x_j|x_{j-1},y_{j-1})
     \prod_{j=1}^s p_{2,\theta}(y_j|x_{j},y_{j-s}^{j-1})\, d\mathfrak m_E^{\otimes s}(x_1^{s-1})$$
and 
$$\forall y_{-s+1}^0, \ \  \int_{E\times K^s} H(x_s,y_{-s+1}^s)\, d\mu(x_s,y_1^s)=1. $$
Therefore, by the Lebesgue dominated convergence theorem, we obtain
$$\lim_{k\rightarrow +\infty}\sup_{i\le 0}
  \int_{E\times K^s}B_kC_i\, d\mu(x_s,y_{1}^s)=0 \ \ \bar{\mathbb P}_{\theta^*}-a.s..$$
Of course, this convergence also holds in $\bar{\mathbb P}_{\theta^*}$-probability.
Now, since, for every $k$, $\tilde A_k$ is a real valued random variable (see (\ref{HH3c})) with the same distribution as
$p_{1,+}^{\ell+1}\prod_{j=s}^{\ell} G(Y_{j-s}^j)\prod_{j=0}^{s-1} \tilde G(Y_j)$, we obtain the result.
\end{proof}
\begin{lemma}\label{lem7}(\cite[Lem. 6 \& 7, Prop. 3]{DMR})
For every $\theta\in\Theta$, $\ell(\theta)\le\ell(\theta^*)$.
Furthermore
$$\ell(\theta)=\ell(\theta^*)\ \ \Rightarrow\ \ \bar{\mathbb P}_{\theta}^{Y}=
      \bar{\mathbb P}_{\theta^*}^{Y}. $$
\end{lemma}
\begin{proof}[Elements of the proof]
We do not rewrite the proof of this lemma, the reader can follow the proofs of
\cite[Lem. 6-7, Prop. 3]{DMR} (using Lemma \ref{lem5} and Kullback-Leibler divergence functions). 
The only adaptations to make concern the proof of \cite[Lem. 7]{DMR} which, 
due to our slightly weaker hypothesis (\ref{HH3c}),
are the following facts. Following the proof of Lemma \ref{lem5}, 
observe that, due to (\ref{HH1}), (\ref{facteq3}) and (\ref{facteq2}), 
on $\{Y_{-s+1}^p=y_{-s+1}^p,\ Y_{-m-s+1}^{-k}=y_{-m-s+1}^{-k}\}$, 
$\bar p_\theta(Y_{-s+1}^p|Y_{-m-s+1}^{-k})$ is between
$$p_{1,-}^{p+s}
   \int_{E\times K^{s}}\prod_{j=-s+1}^pG(y_{j-s}^j)\bar p_\theta(X_{-s}=x_{-s},Y_{-2s+1}^{-s}=
     y_{-2s+1}^{-s}|Y_{-m-s+1}^{-k})\, d\mu(x_{-s},y_{-2s+1}^{-s})$$
and
$$p_{1,+}^{p+s}\int_{E\times K^{s}}\prod_{j=-s+1}^pG(y_{j-s}^j)\bar p_\theta(X_{-s}=x_{-s},Y_{-2s+1}^{-s}=
     y_{-2s+1}^{-s}|Y_{-m-s+1}^{-k})\, d\mu(x_{-s},y_{-2s+1}^{-s}),$$
with $G(y_{-s}^0):=\int_Ep_{2,\theta}(y_0|x,y_{-s}^{-1})\, d\mathfrak m_E(x)$.
Therefore we have
$$\frac{p_{1,-}^{p+s}}{p_{1,+}^s}\prod_{j=1}^pG(Y_{j-s}^j)\le 
        \bar p_\theta(Y_1^p|Y_{-s+1}^0,Y_{-m-s+1}^{-k})=
    \frac{\bar p_\theta(Y_{-s+1}^p|Y_{-m-s+1}^{-k})}{\bar p_\theta(Y_{-s+1}^0|Y_{-m-s+1}^{-k})}
   \le \frac{p_{1,+}^{p+s}}{p_{1,-}^s}\prod_{j=1}^pG(Y_{j-s}^j) .$$
Due to (\ref{HH3}) and (\ref{HH3b}), we obtain
$$\bar{\mathbb E}_{\theta^*}\left[
 \sup_{k}\sup_{m\ge k}|\log(\bar p_\theta(Y_1^p|Y_{-s+1}^0,Y_{-m-s+1}^{-k}))|\right]<\infty,$$
which enables the adaptation of the proof of \cite[Lem. 7]{DMR}.
\end{proof}
\begin{proof}[Proof of Theorem \ref{consistancegene}]
Let $x_0\in E$.
We know that, $\bar{\mathbb P}_{\theta^*}$-almost surely, $(n^{-1}\ell_n(\cdot,x_0))_n$ converges uniformly to $\ell$ 
which admits a maximum $\ell(\theta^*)$. Since $\ell_n(\cdot,x_0)$ is continuous on $\Theta$ and since $\Theta$ is compact, 
$\hat\theta_{n,x_0}$ is well defined. Moreover, the limit values of $(\hat\theta_{n,x_0})_n$ are contained in 
$$\{\theta\in\Theta\ :\ \ell(\theta)=\ell(\theta^*)\}\subseteq\{\theta\in\Theta\ :\ 
           \bar{\mathbb P}_{\theta}^{Y}= \bar{\mathbb P}_{\theta^*}^{Y}\}.$$
Assume now that $Q_{\theta^*}$ 
is aperiodic and positive Harris recurrent, following the proof of \cite[Thm. 5]{DMR},
we have $\lim_{n\rightarrow +\infty}\ell(\hat\theta_{n,x_0})=\ell(\theta^*)$ almost surely for any initial measure
and we conclude as above.
\end{proof}

\section{Identifiability for the Gaussian model: proof of Proposition \ref{identifiabilityAR}}
\label{appendAR}
Assume that $\bar {\mathbb P}_{\theta_1}^Y=\bar {\mathbb P}_{\theta_2}^Y$. In particular, we have 
$$\bar {p}_{\theta_1}(Y_k=y_k|Y_{k-s}^{k-1}=y_{k-s}^{k-1})=\bar {p}_{\theta_2}(Y_k=y_k|Y_{k-s}^{k-1}=y_{k-s}^{k-1}),
\mbox{ for } \ \bar {\mathbb P}_{\theta_1}^{Y_{k-s}^k}-a.e.\ y_{k-s}^k$$
and thus 
$$\sum_{x =1}^ 2 \bar {\mathbb P}_{\theta_1}(X_k=x|y_{k-s}^{k-1}) p_{2,\theta_1}(y_k|x,y_{k-s}^{k-1})
 =\sum_{x =1}^ 2 \bar { \mathbb  P}_{\theta_2}(X_k=x|y_{k-s}^{k-1}) p_{2,\theta_2}(y_k|x,y_{k-s}^{k-1}),$$
for $\bar {\mathbb P}_{\theta_1}^{Y_{k-s}^k}$-almost every $y_{k-s}^k$. 
Since $\bar { p}_{\theta_1}(y_{k-s}^{k}) >0$  (the invariant pdf $h_1$ satisfies $h_1>0$  and the transition pdf $q_\theta$ satisfies 
$q_\theta>0$ by construction), this last equality also holds for Lebesgue almost every $y_{k-s}^k$.
According to  \cite{Teich63}, finite mixtures of Gaussian distribution are identifiable. 
Due to \eqref{eq:AR1}, this implies in particular that if 
$$\sum_{x=1}^2 \pi_x^{(1)} \mathcal N(y;a^{(1)}_x,\sigma^{(1)}_x) = \sum_{x=1}^M \pi_x^{(2)} \mathcal N(y;a^{(2)}_x,\sigma^{(2)}_x)\mbox{ for }- a.e.\  y$$
with  $(a^{(1)}_1,\sigma^{(1)}_1)\neq (a^{(1)}_2,\sigma^{(1)}_2)$, $\pi_1^{(1)}>0$ and $\pi_2^{(1)}>0$, then there exists 
a permutation $\tau : \{1,2\} \rightarrow \{1,2\}$ such that 
$(a^{(1)}_x,\sigma^{(1)}_x)=(a^{(2)}_{\tau(x)},\sigma^{(2)}_{\tau(x)})$ and $\pi_1^{(x)}=\pi_2^{(\tau(x))}$.
Therefore, since for every $x \in \{1,2\}$ and for Lebesgue almost every $y_{k-s}^{k-1}$, 
$\bar {\mathbb P}_{\theta_1}(X_k=x|y_{k-s}^{k-1})>0$  (since $h_{\theta_1}>0$), 
for Lebesgue almost every $y_{k-s}^{k-1}$
there exists a permutation $\tau_y=\tau_{y_{k-s}^{k-1}}$ of $\{1,...,M\}$ such that,
$$\forall x\in\{1,2\},\ \ 
\left(\beta_{0,(1)}^{(x)}+\sum_{\ell=1}^s\beta_{\ell,(1)}^{(x)}y_{k-\ell},\sigma_{(1)}^{(x)}\right)=
\left(\beta_{0,(2)}^{(\tau_{y}(x))}+\sum_{\ell=1}^s\beta_{\ell,(2)}^{(\tau_{y}(x))}y_{k-\ell},\sigma_{(2)}
^{(\tau_{y}(x))}\right).$$
Recall that we have assumed (for the first model)
$$
\left(\beta_{0,(1)}^{(1)},\beta_{1,{(1)}}^{(1)},... ,\beta_{s,{(1)}}^{(1)}, \sigma^{(1)}_{(1)} \right) \ne \left(\beta_{0,(1)}^{(2)},\beta_{1,{(1)}}^{(2)},... ,\beta_{s,{(1)}}^{(2)}, \sigma^{(2)}_{(1)} \right)
$$

which implies
$$\left(\beta_{0,(1)}^{(1)}+\sum_{\ell=1}^s\beta_{\ell,(1)}^{(1)}y_{k-\ell},\sigma_{(1)}^{(1)}\right)\ne 
\left(\beta_{0,(1)}^{(2)}+\sum_{\ell=1}^s\beta_{\ell,(1)}^{(2)}y_{k-\ell},\sigma_{(1)}^{(2)}\right),$$
for Lebesgue almost every $y_{k-s}^{k-1}$. 
Since the set of permutations of $\{1,...,M\}$ is finite, there exists a positive Lebesgue
measure subset of $\mathbb T^s$ on which
the permutation is the same permutation $\tau$.
From this, we deduce that, for all  $x\in \{1,2\}$ and $y \in \mathbb R$, 
$$\left(\beta_{0,(1)}^{(x)},\beta_{1,(1)}^{(x)},... ,\beta_{r,(1)}^{(x)}, \sigma^{(x)}_{(1)} \right)=\left(\beta_{0,(2)}^{(\tau(x))},\beta_{1,(2)}^{(\tau(x))},... ,\beta_{r,(2)}^{(\tau(x))}, \sigma^{(\tau(x))}_{(2)} \right)$$ and
\begin{eqnarray}
\label{eq:identp1}
p_{1,\theta_{1}} (x|x,y)
&=&\pi_{-,(1)}^{(x)} +\frac{1-\pi_{-,(1)}^{(x)}-\pi_{+,(1)}^{(x)}}{1+exp\left(\lambda_{0,(1)}^{(x)}+\lambda_{1,(1)}^{(x)} y\right)}\\
&=&\pi_{-,(2)}^{(\tau(x))} +\frac{1-\pi_{-,(2)}^{(\tau(x))}-\pi_{+,(2)}^{(\tau(x))}}{1+exp\left(\lambda_{0,(2)}^{(\tau(x))}+\lambda_{1,(2)}^{(\tau(x))} y\right)} 
= p_{1,\theta_2}(\tau(x)|\tau(x),y).\nonumber
\end{eqnarray}
If $\theta_1$ and $\theta_2$ are in $\Theta'$ then  $\lambda_{1,(i)}^{(x)} \neq 0$ for $i \in \{1,2\}$ and looking at the asymptotic behavior of the terms which appear in (\ref{eq:identp1}) when $y \rightarrow \pm \infty$ permits to show that $\pi_{-,(1)}^{(x)}=\pi_{-,(2)}^{(x)}$, $\pi_{+,(1)}^{(x)}=\pi_{+,(2)}^{(x)}$. We can then easily deduce that $\lambda_{0,(1)}^{(x)}=\lambda_{0,(2)}^{(x)}$ and $\lambda_{1,(1)}^{(x)}=\lambda_{1,(2)}^{(x)}$ and thus that $\theta_1=\theta_2$.

If $\theta_1$ and $\theta_2$ are in $\Theta''$, then we directly obtain that 
$\pi_{-,(1)}^{(x)} = \pi_{-,(2)}^{(x)}=\pi_{+,(1)}^{(x)} =\pi_{-,(1)}^{(x)}=\pi_0$
and then that $\theta_1=\theta_2$.\qed
\section{Identifiability for the Rainfall model: proof of Proposition \ref{identrain}}\label{proofidentrain}
Assume that $\bar{\mathbb P}_{\theta_1}^Y=\bar{\mathbb P}_{\theta_2}^Y$.
First, we use the fact that
\begin{equation}\label{EE1}
\bar p_{\theta_1}(Y_k=y_k|Y_{k-s}^{k-1}=y_{k-s}^{k-1})=
\bar p_{\theta_2}(Y_k=y_k|Y_{k-s}^{k-1}=y_{k-s}^{k-1})\ \ \mbox{for}\ \ 
\bar{\mathbb P}_{\theta_1}^{Y_{k-s}^k}-a.e.\ y_{k-s}^k
\end{equation}
to prove that 
$$(\pi_{i,(1)}^{(x)},\alpha_{i,(1)}^{(x)},\beta_{i,(1)}^{(x)})_
{i,x}=(\pi_{i,(2)}^{(x)},\alpha_{i,(2)}^{(x)},\beta_{i,(2)}^{(x)})_{i,x}.$$
Using (\ref{EE1}) on the set $\{r_k^{(i)}>0,\ \forall i\in\{1,...,\ell\}\}$,
we conclude that there exists a permutation $\tau$ of $\{1,...,M\}$
such that, for every $i\in\{1,...,\ell\}$ and every $x\in\{1,...,M\}$,
we have
\begin{equation}\label{EE1b}
(\alpha_{i,(1)}^{(x)},\beta_{i,(1)}^{(x)})=(\alpha_{i,(2)}^{(\tau(x))},
\beta_{i,(2)}^{(\tau(x))})
\end{equation}
and
$$\bar{\mathbb P}_{\theta_1}(X_k=x|Y_{k-s}^{k-1}=y_{k-s}^{k-1})\prod_{i=1}^\ell
   \pi_{i,(1)}^{(x)}
=\bar{\mathbb P}_{\theta_2}(X_k=x|Y_{k-s}^{k-1}=y_{k-s}^{k-1})\prod_{i=1}^\ell
   \pi_{i,(2)}^{(\tau(x))}.$$
Now, for every $J\subseteq\{1,...,\ell\}$, we use 
(\ref{EE1}) on the set $\{r_k^{(j)}>0,\ \forall j\in J,\ r_k^{(i)}=0,\
  \forall i\not\in J\}$. Due to (\ref{EE1b}) and since $\theta_1$
satisfies (\ref{condiident}), we obtain
$$\bar{\mathbb P}_{\theta_1}(X_k=x|Y_{k-s}^{k-1}=y_{k-s}^{k-1})\prod_{j\in J}
   \pi_{j,(1)}^{(x)}\prod_{i\not\in J}(1-\pi_{i,(1)}^{(x)})
=\bar{\mathbb P}_{\theta_2}(X_k=x|Y_{k-s}^{k-1}=y_{k-s}^{k-1})\prod_{j\in J}
   \pi_{j,(2)}^{(\tau(x))}\prod_{i\not\in J}(1-\pi_{i,(2)}^{(\tau(x))}).$$
From which, we conclude
\begin{equation}\label{EE1c}
\forall i\in\{1,...,\ell\},\ \forall x\in\{1,...,M\},\ \ \pi_{i,(1)}^{(x)}
=\pi_{i,(2)}^{(\tau(x))}.
\end{equation}
Now it remains to prove that $(q_{x,x',(1)},\mu_{x,x',(1)})=(q_{\tau(x),\tau(x'),(2)},\mu_{\tau(x),\tau(x'),(2)})$. To this hand, as for the AR model (see Appendix \ref{appendAR}), we use the fact that
\begin{equation}\label{EE2}
\bar p_{\theta_1}(Y_k=y_k,Y_{k+1}=y_{k+1}|Y_{k-s}^{k-1}=y_{k-s}^{k-1})=
\bar p_{\theta_2}(Y_k=y_k,Y_{k+1}=y_{k+1}|Y_{k-s}^{k-1}=y_{k-s}^{k-1})\ \ \mbox{for}\ \ 
\bar{\mathbb P}_{\theta_1}^{Y_{k-s}^k}-a.e.\ y_{k-s}^k
\end{equation}
and obtain that
$$\forall x,x',\ \ p_{1,\theta_1}(x'|x,y_k)=p_{1,\theta_2}(\tau(x')|\tau(x),y_k)\ \mbox{for}\ a.e. y_k .$$
This implies that
\begin{equation}\label{EE2b}
\frac{\tilde q_{x,x',(1)}\exp(-z_{k-1}'\tilde\mu_{x,x',(1)})}
  {\sum_{x"}\tilde q_{x,x",(1)}\exp(-z_{k-1}'\tilde\mu_{x,x',(1)})}=\frac{\tilde q_{\tau(x),\tau(x'),(2)}\exp(-z_{k-1}'\tilde\mu_{\tau(x),\tau(x'),(2)})}
  {\sum_{x"}\tilde q_{x,\tau(x"),(2)}\exp(-z_{k-1}'\tilde\mu_{x,\tau(x"),(2)})},
\end{equation}
with
$\tilde q_{x,x',(j)}:=q_{x,x',(j)}\exp(-\frac 12(\mu_{x,x',(j)})'\Sigma^{-1}
  \mu_{x,x',(j)})$ and $\tilde \mu_{x,x',(j)}:=\Sigma^{-1}\mu_{x,x',(j)}$.
From (\ref{EE2b}), we obtain that
$$
\frac{\tilde q_{x,x',(1)}\exp(- z_{k-1}'\tilde\mu_{x,x',(1)})}
  {\tilde q_{x,x,(1)}\exp(- z_{k-1}'\tilde\mu_{x,x,(1)})}=\frac{\tilde q_{\tau(x),\tau(x'),(2)}\exp(-z_{k-1}'\tilde\mu_{\tau(x),\tau(x'),(2)})}
  {\tilde q_{\tau(x),\tau(x),(2)}\exp(-z_{k-1}'\tilde\mu_{\tau(x),\tau(x),(2)})},
$$
and so that, for every $x,x'\in\{1,...,M\}$, 
\begin{equation}\label{EE2c}
\tilde\mu_{x,x',(1)}-\tilde\mu_{x,x,(1)}=\tilde\mu_{\tau(x),\tau(x'),(2)}-\tilde\mu_{\tau(x),\tau(x),(2)}
\end{equation}
and
\begin{equation}\label{EE2d}
\frac{\tilde q_{x,x',(1)}}{\tilde q_{x,x,(1)}}=
\frac{\tilde q_{\tau(x),\tau(x'),(2)}}{\tilde q_{\tau(x),\tau(x),(2)}}.
\end{equation}
Finally, it comes from (\ref{EE2c}) that $\tilde\mu_{x,x',(1)}=\tilde\mu_{\tau(x),\tau(x'),(2)}$
(using $\sum_{x"}\tilde\mu_{x,x",(j)}=0$) and so $\mu_{x,x',(1)}=\mu_{\tau(x),\tau(x'),(2)}$.
So (\ref{EE2d}) becomes
$$\frac{ q_{x,x',(1)}}{ q_{x,x,(1)}}=
\frac{ q_{\tau(x),\tau(x'),(2)}}{ q_{\tau(x),\tau(x),(2)}}$$
which implies that $q_{x,x',(1)}= q_{\tau(x),\tau(x'),(2)}$
(due to $\sum_{x"}q_{x,x",(j)}=1$).\qed
\section{Proof of Lemma \ref{propHH4}}\label{proofrain}
Let $f$ be any probability pdf wrt $\mu=\mathfrak m_E\times \mathfrak m_K$. We have
\begin{eqnarray*}
 [Q_\theta^{*n}(f-h_\theta)](x_0,y_0)&=&
     \int_{(E\times K)^{n}}\prod_{i=-n+1}^0q_\theta(x_i,y_i|x_{i-1})(f-h_\theta)(x_{-n},y_{-n})
       \, d\mathfrak m_E^{\otimes n}(x_{-n}^{-1})d\mathfrak m_K^{\otimes n}(y_{-n}^{-1})\\
  &=&     \int_{E^n\times K^{n-1}}\prod_{i=-n+1}^0q_\theta(x_i,y_i|x_{i-1})(F-h_{1,\theta})(x_{-n})
       \, d\mathfrak m_E^{\otimes n}(x_{-n}^{-1})d\mathfrak m_K^{\otimes (n-1)}(y_{-n+1}^{-1})\\
\end{eqnarray*}
with $F(x_{-n}):=\int_Kf(x_{-n},y_{-n})\, d\mathfrak m_K(y_{-n})$.
Now, since $q_\theta(x_i,y_i|x_{i-1})=p_{1,\theta}(x_i|x_{i-1})p_{2,\theta}(y_i|x_i)$, we obtain that
$$ [Q_\theta^{*n}(f-h_{1,\theta})](x_0,y_0)=
   p_{2,\theta}(y_0|x_0) \int_{E^n}\prod_{i=-n+1}^0p_{1,\theta}(x_i|x_{i-1})(F-h_{1,\theta})(x_{-n})
       \, d\mathfrak m_E^{\otimes n}(x_{-n}^{-1}).$$
Therefore
$$||Q_\theta^{*n}(f-h_\theta) ||_{L^1(\mathfrak m_E\times \mathfrak m_K)}= ||Q_{1,\theta}^{*n}(F-h_{1,\theta}) ||_{L^1(\mathfrak m_E)} .$$

Now, let us assume that $p_{2,\theta}>0$ and that $(X_k)_k$ is an aperiodic positive Harris recurrent Markov chain.
We will use the notations of \cite{MeynTweedie}.

Since $(X_k)_k$ is positive, it is $\psi$-irreducible (with $\psi=\psi_0$). Due to the hypothesis on $p_{2,\theta}$,
this implies the $\psi$-irreducibility of $(X_k,Y_k)_k$ (with $\psi=\psi_0\times \mathfrak m_K$).

Moreover $(X_k,Y_k)_k$ is positive since it admits an invariant probability measure (due to the first point of this result).

The fact that $(X_k)_k$ is aperiodic means that, for every $\nu_M$-small set $C$ such that $\nu_M(C)>0$ for $(X_k)_k$, 
the greatest common divisor of the set $E_C$ defined as follows is equal to 1:
$$E_C:=\{n\ge 1\ :\ C\mbox{ is }\nu_n-\mbox{small with }\nu_n=\delta_n\nu_M\mbox{ and }\delta_n>0\}.$$

Now, let $C'$ be a $\nu'_M$-small set for $(X_k,Y_k)_k$ with $\nu'_M(C')>0$, then for every $(x_0,y_0)\in C'$
and every $(B,D)\in\mathcal B(E)\times \mathcal B(K)$, we have $Q_\theta^M\ind_{B\times D}(x_0,y_0)\ge \nu'_M(B\times D)$.
Moreover $Q_\theta^M\ind_{B\times D}(x_0,y_0)$ is equal to
$$\int_{E^{M-1}}\left(\int_B\prod_{i=1}^Mp_{1,\theta}(x_i|x_{i-1})\left(\int_D
         p_{2,\theta}(y_M|x_M)\, d\mathfrak m_K(y_M)\right)\, d\mathfrak m_E(x_M)\right)\, d\mathfrak m_E^{\otimes (M-1)}(x_1^{M-1}).$$
Since $Q_\theta^M\ind_{B\times D}(x_0,y_0)$ does not depend on $y_0$, we obtain
$$ \forall (x_0,y_0)\in E\times K,\ \forall B\in\mathcal B(E),\  \ 
    Q_{1,\theta}^M\ind_{B}(x_0)=Q_\theta^M\ind_{B\times K}(x_0,y_0)\ge \nu'_M(B\times K)$$
and so
$C:=\{x\in E\ :\ \exists y\in K,\ (x,y)\in C'\}$ is $\nu_M$-small with
$\nu_M(B)=\nu'_M(B\times K)$ and $\nu_M(C)\ge \nu'_M(C')>0$. Moreover $E_C=E_{C'}$.
Indeed, if $C'$ is $\nu'_n$-small with $\nu'_n=\delta'_n\nu'_M$, then
$C$ is $\nu_n$-small with $\nu_n(B)=\nu'_n(B\times K)=\delta_n\nu_M(B)$ with $\delta_n(x)=\int_K\delta'_n(x,y)\, d\mathfrak m_K(y)$;
and conversely, if $C$ is $\nu_n$-small with $\nu_n=\delta_n\nu_M$, then
$C'$ is $\nu'_n$-small with $\nu'_n(B\times D)=\delta'_n\nu'_M(B\times D)$ and with $\delta'_n(x,y)=\delta_n(x)p_{2,\theta}(y|x)$.
Therefore $(X_k,Y_k)_k$ is also aperiodic.

Finally, the Harris recurrence property of $(X_k,Y_k)_k$ follows from the Harris-recurrence of $(X_k)_k$
and from $p_{2,\theta}>0$.\qed

\end{appendix}

\bibliographystyle{plain}

\bibliography{bib2}

\end{document}